\documentclass[technologies,article,accept,pdftex,moreauthors]{Definitions/mdpi}

\usepackage{graphicx}%
\usepackage{multirow}%
\usepackage{amsmath,amssymb,amsfonts}%
\usepackage{amsthm}%
\usepackage{mathrsfs}%
\usepackage[title]{appendix}%
\usepackage{xcolor}%
\usepackage{textcomp}%
\usepackage{manyfoot}%
\usepackage{booktabs}%
\usepackage{algorithm}%
\usepackage{algorithmicx}%
\usepackage{algpseudocode}%
\usepackage{listings}%
\usepackage{url}
\usepackage{rotating}

\definecolor{em}{rgb}{1,0,0}

\theoremstyle{plain}
\newtheorem{thm}{Theorem}
\newtheorem{lem}[theorem]{Lemma}

\firstpage{1} 
\pubvolume{1}
\issuenum{1}
\articlenumber{0}
\pubyear{2023}
\copyrightyear{2023}
\externaleditor{Academic Editor: Firstname Lastname}
\datereceived{4 August 2023} 
\daterevised{1 September 2023} 
\dateaccepted{12 October 2023} 
\datepublished{ } 
\hreflink{https://doi.org/} 

\Title{Dual-Matrix Domain Wall: A Novel Technique for Generating Permutations by QUBO and Ising Models with Quadratic Sizes} 

\TitleCitation{Dual-Matrix Domain Wall: A Novel Technique for Generating Permutations by QUBO and Ising Models with Quadratic Sizes}

\Author{
{Koji} 
 Nakano $^{1,}$*\orcidA{},
Shunsuke Tsukiyama $^{1}$, 
Yasuaki Ito $^{1}$\orcidB{},
Takashi Yazane $^{2}$,
Junko Yano $^{2}$,
Takumi Kato $^{2}$,
\linebreak Shiro Ozaki $^{2}$,
Rie Mori $^{2}$
and Ryota Katsuki $^{2}$}

\AuthorNames{
Koji Nakano, 
Shunsuke Tsukiyama,
Yasuaki Ito,
Takashi Yazane,
Junko Yano,
Takumi Kato,
Shiro Ozaki,
Rie Mori
and Ryota Katsuki}

\AuthorCitation{Nakano, K.; Tsukiyama, S.; Ito, Y.; Yazane, T.; Yano, J.; Kato, T.; Ozaki, S.; Mori, R.; Katsuki, R.}

\address{%
$^{1}$ \quad Graduate School of Advanced Science and Engineering, Hiroshima University, Kagamiyama 1-4-1, Higashihiroshima 739-8527, {Hiroshima,} 
  Japan; {tsukiyama@cs.hiroshima-u.ac.jp (S.T.); yasuaki@cs.hiroshima-u.ac.jp (Y.I.)} 
\\
$^{2}$ \quad Research and Development Headquarters, NTT DATA Group Corporation, Toyosu Center Bldg, Annex, 3-9, Toyosu 3-chome, {Koto-ku} 
 135-8671, {Tokyo,} Japan; {takashi.yazane@nttdata.com (T.Y.); \linebreak junko.yano@nttdata.com (J.Y.); takumi.kato@nttdata.com (T.K.); shiro.ozaki@nttdata.com (S.O.); rie.mori@nttdata.com (R.M.); ryota.katsuki@nttdata.com (R.K.)}}

\corres{Correspondence: nakano@cs.hiroshima-u.ac.jp}

\abstract{The Ising model is defined by an objective function using a quadratic formula of qubit variables. The problem of an Ising model aims to determine the qubit values of the variables that minimize the objective function, and many optimization problems can be reduced to this problem.
In this paper, we focus on optimization problems related to permutations, where the goal is to find the optimal permutation out of the $n!$ possible permutations of $n$ elements. To represent these problems as Ising models, a commonly employed approach is to use a kernel that applies one-hot encoding to find any one of the $n!$ permutations as the optimal solution. However, this kernel contains a large number of quadratic terms and high absolute coefficient values.
The main contribution of this paper is the introduction of a novel permutation encoding technique called the dual-matrix domain wall, which significantly reduces the number of quadratic terms and the maximum absolute coefficient values in the kernel. Surprisingly, our dual-matrix domain-wall encoding reduces the quadratic term count and maximum absolute coefficient values from $n^3-n^2$ and $2n-4$ to $6n^2-12n+4$ and $2$, respectively.
We also demonstrate the applicability of our encoding technique to partial permutations and Quadratic Unconstrained Binary Optimization (QUBO) models. Furthermore, we discuss a family of permutation problems that can be efficiently implemented using Ising/QUBO models with our dual-matrix domain-wall encoding.
}


\begin{document}

\section{Introduction}
\emph{{A Binary Quadratic Model (BQM)} 
}~\cite{dimod-BQM} is defined by an objective function that includes a quadratic formula with multiple variables.
The problem associated with a BQM is to find the values of these variables that minimize the resulting value of the quadratic formula.
A BQM is referred to as \emph{a Quadratic Unconstrained Binary Optimization (QUBO)}~\cite{Kochenberger14} model when the variables are restricted to binary values, i.e., they can only take \emph{bit (or binary)} values in $\{0,1\}$.
On the other hand, if the variables can only take \emph{qubit (or spin)} values in $\{-1,+1\}$, the model is called an Ising model.
It is worth noting that QUBO and Ising models can be converted into each other interchangeably~\cite{dimod-BQM,Tao20}.

\textls[-10]{Since optimization problems such as the traveling salesman problem; scheduling problems; and various graph problems, including max cut, maximum independent set, and graph isomorphism, can be transformed into QUBO/Ising models~\cite{Lucas14},
there has been significant research dedicated to finding efficient algorithms, hardware, and systems to solve them.
However, solving optimization problems for QUBO/Ising models is known to be NP-hard.
This means that unless $P=NP$, it is not possible to design a polynomial time algorithm using classical computers with digital circuit devices of polynomial size.
In the quest for solutions, researchers have explored the potential of ideal quantum annealers based on quantum mechanics, which could potentially find optimal solutions for large Ising models within a reasonable time frame~\cite{Kadowaki98}.
Unfortunately, the current quantum annealers available are not yet powerful enough to tackle such problems effectively.
The number of qubits is limited, and the presence of undesirable flux noise significantly reduces the probability of finding optimal solutions~\cite{Zaborniak21}.
Hence, as an alternative to an ideal quantum annealer, BQM solvers
on various non-quantum computing platforms, such as ASICs~\cite{Oku19,Yamamoto21}, 
FPGAs~\cite{Goto21,Goto19,Kagawa21,Matsubara17},
GPUs~\cite{Okuyama19,Yasudo-JPDC22,Yasudo-APDCM22,Nakano23}, and optimal fibers~\cite{Honjo21,Inagaki16}, have been proposed.
Further, D-Wave Systems released a hybrid BQM solver~\cite{D-Wave-Hybrid20} that
uses both a classical computer and a quantum annealer
to find solutions for large BQMs with up to 1,000,000-node complete graphs.}

In this paper, our primary focus is on the size and resolution of QUBO/Ising models, which are represented by the quadratic term count and the maximum absolute value of the coefficients, respectively.
These metrics serve as important indicators for assessing the models.
While it is not always the case, QUBO/Ising models with smaller sizes and resolutions are generally considered more desirable.
This is particularly true when considering the limitations of quantum annealers, which have restricted size and resolution capabilities.
Therefore, it becomes crucial to ensure that the Ising models embedded in quantum annealers possess small sizes and low resolutions.
 Additionally, even when these models are processed by classical digital computers, the required memory size to store QUBO/Ising models is proportional to the quadratic term count, and a higher resolution requires a larger word size of memory. 
With this perspective in mind, this paper places significant emphasis on the size and resolution of QUBO/Ising models.
We provide precise evaluations of these metrics, offering valuable insights into their characteristics.

QUBO/Ising models that are converted from permutation-based combinatorial optimization problems for $n$ elements should have a kernel capable of generating any one of the $n!$ possible permutations as the optimal solution.
To achieve this, a common approach is to use the one-hot encoding of permutations, which involves a $bit/qubit$ matrix of size $n\times n$.
In this encoding, each row $i$ ($0\leq i\leq n-1$) of the matrix represents the $i$-th number as a one-hot vector, where exactly one element is set to $1$/$+1$, and its position corresponds to the number it represents.
{The reader should refer to Figure~\ref{fig:onehot}, which illustrates a $4\times 4$ matrix representing the permutation $[1,3,2,0]$.}
In order to generate any one of these matrices as the optimal solution, QUBO/Ising models require kernels that have optimal solutions if and only if each row and each column contains exactly one $1$/$+1$.
{For instance, QUBO models utilizing this kernel have been proposed for addressing Hamiltonian cycle/path problems and the graph isomorphism problem, as demonstrated in~\cite{Lucas14}. Similarly, models for the traveling salesman problem (TSP) and the quadratic assignment problem (QAP) were introduced in~\cite{Ayodele22,Goh22}.}
However, these kernels typically involve $n^3-n^2$ quadratic terms.
Additionally, the kernel of the Ising model features a non-constant large coefficient of $2n-4$.
Recently, a permutation generation technique using domain-wall encoding has been introduced~\cite{Codognet22}.
This technique employs a matrix of size $n\times (n-1)$, where each row $i$ ($0\leq i\leq n-1$) stores the $i$-th number as a domain-wall vector~\cite{Chancellor19,Chen21,Berwald22}.
By utilizing this technique, the number of quadratic terms is reduced to ${1\over 2}n^3-{3\over 2}n$, which is half the number of quadratic terms in QUBO/Ising kernels obtained through conventional one-hot encoding.
This technique still involves a cubic number of quadratic terms.
Moreover, QUBO/Ising models using this domain-wall encoding require coefficients with large absolute values, namely $2n-3$/$n-1$, respectively.

The main contribution of this paper is to introduce a novel permutation encoding technique called \emph{the dual-matrix domain wall},
which uses two matrices $A$ and $B$ to store dual permutations.
This new technique can significantly reduce the number of quadratic terms and the maximum absolute coefficient values in the resulting Ising kernel.
The QUBO/Ising kernel obtained by this technique has only $6n^2-12n+4$ quadratic terms.
Also, the maximum absolute value of the coefficients of the Ising kernel is just $2$.

\begin{figure}[H]
\includegraphics[scale=0.6]{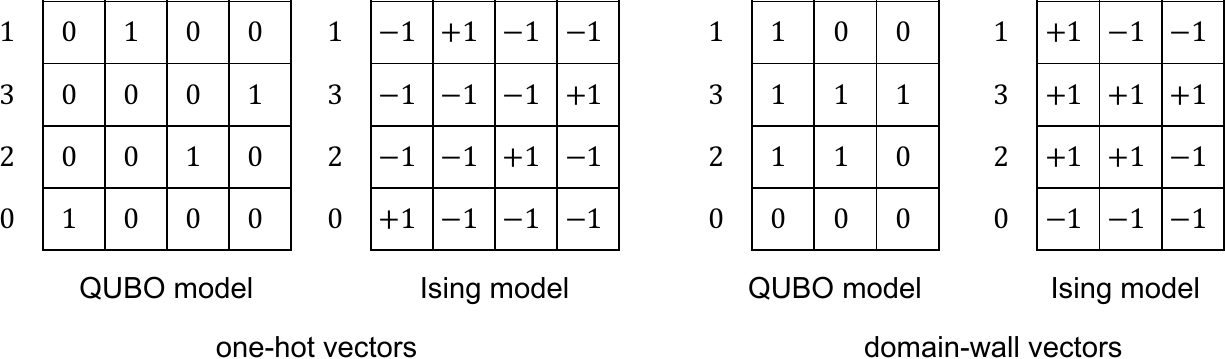}
\caption{Representation of a permutation $[1,3,2,0]$ by a $4\times 4$ matrix with one-hot vectors and a $4\times 3$ matrix
with domain-wall vectors}
\label{fig:onehot}
\end{figure}

We also discuss partial-permutation generation by QUBO/Ising kernels,
which is a sequence of $m$ numbers selected from $n$ numbers 0, 1, $\ldots$, $n-1$ without repetition.
It is not difficult to apply one-hot encoding to generate a partial permutation~\cite{Calude17,Yoshimura21}.
For partial-permutation generation, one uses 
an $m\times n$ matrix such that each $i$-row  \mbox{($0\leq i\leq n-1$)} 
stores the $i$-th selected number as a one-hot vector.
QUBO/Ising kernels that use one-hot encoding to generate any such matrices need ${1\over 2}m^2n+{1\over 2}mn^2-mn$ quadratic terms.
We show that we can apply the dual-matrix domain-wall technique for partial-permutation generation,
and the resulting QUBO/Ising models have only $6mn-6m-6n+4$ quadratic terms.

Moreover, we introduce a generic problem called the particle placement problem (PPP), 
which aims to optimize the placement of $m$ particles in $n$ locations ($m\leq n$) with no collision.
In other words, the problem is to find an optimal permutation of $m$ numbers selected from $n$ numbers.
In this study, we demonstrate that the PPP can be efficiently reduced to a QUBO/Ising model using permutation-generating kernels.
Additionally, we establish that several permutation-based combinatorial optimization problems, such as the quadratic assignment problem (QAP), the traveling salesman problem (TSP), the sub-graph isomorphism problem, and the maximum weight matching problem, can also be reduced to the PPP.
Thus, these problems can be converted to equivalent QUBO/Ising models using permutation-generating kernels.

Finally, we conducted an evaluation of the cells required to embed Ising kernels on a D-Wave quantum annealer, specifically the Advantage 4.1.
The objective was twofold: first, to assess the feasibility of utilizing the dual-matrix domain-wall technique on presently accessible quantum annealers, and second, to examine the influence of the number of quadratic terms in Ising models on the cell requirements.

This paper is organized as follows.
In Section~\ref{sec:pre}, we begin by providing a formal definition of QUBO and Ising models. We also establish the relationship between quantum annealers and Ising models.
Additionally, we explore different encoding techniques such as one-hot, zero-one-hot, and domain-wall encoding for representing numbers in QUBO and Ising models.
Section~\ref{sec:conventional} reviews the conventional one-hot encoding technique used to represent permutations of $n$ numbers.
Furthermore, we introduce the all-different domain-wall encoding technique, which effectively reduces the number of quadratic terms by half~\cite{Codognet22}.
However, both of these techniques require a cubic number of quadratic terms.
In Section~\ref{sec:dual-matrix}, we present a novel encoding technique called \emph{{dual-matrix domain-wall encoding}
}.
This technique enables the generation of permutations using a QUBO/Ising kernel with only a quadratic number of quadratic terms.
{Section}~\ref{sec:partial} 
generalizes the concept of permutations to partial permutations, which represent a partial permutation of $m$ numbers selected from a set of $n$ numbers without repetition.
We demonstrate how QUBO/Ising kernels can generate partial permutations using both the conventional one-hot encoding technique and our dual-matrix domain-wall encoding technique.
In Section~\ref{sec:extended}, we extend the dual-matrix domain-wall technique by incorporating a matrix called a one-hot matrix that stores a partial permutation as a one-hot encoding.
We introduce the particle placement problem (PPP) in Section~\ref{sec:ppp}, highlighting its ability to reduce many permutation-based combinatorial optimization problems.
We evaluate the number of quadratic terms in the resulting QUBO/Ising models obtained through reduction via the PPP for each problem.
Section~\ref{sec:embedding} presents an analysis of the number of cells required to embed Ising kernels for generating permutations in a D-Wave quantum annealer.
Finally, Section~\ref{sec:concl} concludes \mbox{our work.}

\section{Preliminaries}\label{sec:pre}
The main objective of this section is to introduce the concepts of QUBO and Ising models, as well as the fundamental aspects of their design.
To begin with, we provide a formal definition of Ising and QUBO models.
Subsequently, we establish the connections between Ising models and quantum annealers, offering valuable insights to the reader.
Lastly, we showcase the applications of QUBO and Ising models in generating one-hot, zero-one-hot, and domain-wall vectors.

\subsection{QUBO and Ising Models}
\emph{{A Binary Quadratic Model (BQM)}}~\cite{dimod-BQM} is defined as an objective function with a quadratic formula involving multiple variables.
The goal of a BQM is to determine the variable values that minimize the resulting value of the quadratic formula.

A model is referred to as \emph{a Quadratic Unconstrained Binary Optimization (QUBO) model}~\cite{Kochenberger14} if the variables can take \emph{bit (or binary)} values of 0 or 1.
Specifically, let $X=(x_i)$ ($0\leq i\leq n-1$) represent an $n$-bit vector of binary variables.
A QUBO model with $X$ can be defined using an upper triangular weight matrix $W=(W_{i,j})$ ($0\leq i\leq j\leq n-1$).
The objective of the QUBO problem is to find the binary values of $X$ that minimize the energy $E(X)$, which is defined by the following quadratic formula:
\begin{align}
E(X)  &= \sum_{i=0}^{n-1}\sum_{j=i+1}^{n-1}W_{i,j}x_ix_j+\sum_{i=0}^{n-1}W_{i,i}x_i+C, \label{ex:1}\\
        &= \sum_{i=0}^{n-1}\sum_{j=i}^{n-1}W_{i,j}x_ix_j+C\label{ex:2}.
\end{align}

Here, $C$ is a constant called the {\em offset}.
Note that Equations~(\ref{ex:1}) and (\ref{ex:2}) are equivalent because $x_i^2=x_i$ always holds.
While most papers use the energy $E(X)$ in Equation~(\ref{ex:2}) with no offset $C$, 
this paper adopts the energy $E(X)$ defined by Equation~(\ref{ex:1}) to distinguish linear and quadratic terms with coefficients
$W_{i,i}$ ($0\leq i\leq n-1$) and $W_{i,j}$ ($0\leq i<j\leq n-1$), respectively.

A BQM is called \emph{an Ising model} if variables take \emph{qubit (or spin)} values of $-1$ or $+1$.
An Ising model with an $n$-qubit vector $S=(s_i)$ ($0\leq i\leq n-1$) is defined
by quadratic term coefficients $J=(J_{i,j})$ ($0\leq i<j\leq n-1$) called \emph{interactions} and linear term coefficients $h=(h_i)$ ($0\leq i\leq n-1$)
called \emph{biases}.
The Ising problem aims to find the qubit values of $S$ that minimize \emph{the Hamiltonian} $H(S)$, which is defined by the following quadratic formula:
\begin{align}
H(S) &= \sum_{i=0}^{n-2}\sum_{j=i+1}^{n-1}J_{i,j}s_is_j+\sum_{i=0}^{n-1}h_is_i+C.
\end{align}

We make the assumption that all coefficients in the linear and quadratic terms of QUBO/Ising models are integers, unless otherwise specified.
Using integers with large absolute values can lead to discrete values below the effective resolution, so it is important to keep the maximum absolute value of all coefficients as small as possible.
To achieve this, we design QUBO/Ising models that have no common factor in all coefficients.
This allows us to reduce the coefficients by dividing the common factor without affecting the optimal solutions.
It is worth mentioning that the offset $C$ does not influence the optimal solution and can be disregarded. Consequently, it can take a non-integer value.

We define the number of non-zero elements in $W_{i,i}$/$h_i$ ($0\leq i\leq n-1$) as \emph{the linear term count}
and that in $W_{i,j}$/$J_{i,j}$  ($0\leq i<j\leq n-1$) as \emph{the quadratic term count} of a QUBO/Ising model.
Generally, having smaller term counts is advantageous as it helps reduce the hardware resource usage of quantum computers when solving  problems of QUBO/Ising models.
In particular, the quadratic term count has a significant impact on the hardware resource usage of quantum annealers when solving Ising problems.
By minimizing the quadratic term counts, we can effectively reduce the amount of resources needed for implementing and solving QUBO/Ising models on quantum hardware.
This optimization is essential for improving the efficiency and performance of quantum computing systems.

It is easy to see that QUBO and Ising models can be equivalently converted to each other \cite{Tao20}.
Both QUBO and Ising models, excluding the offset $C$, can be represented as weighted undirected graphs
with $n$ nodes $0, 1, \ldots, n-1$
such that $W_{i,i}$/$h_i$ is the weight of node $i$ and
$W_{i,j}$/$J_{i,j}$ is the weight of edge $(i,j)$ of the graph.
Figure~\ref{fig:qubo-ising} shows a graph corresponding to a QUBO model with the
the following energy:
\begin{align}
E(X) &= 2x_0x_1-x_0x_3-2x_0x_4-x_1x_2+x_1x_5+3x_2x_5-2x_3x_4-2x_4x_5 \nonumber \\
       & \qquad   -2x_0-x_1+2x_2+4x_5
\end{align}

The figure also illustrates the equivalent Ising model with the following Hamiltonian:
\begin{align}
H(S) &=  2s_0s_1-s_0s_3-2s_0s_4  -s_1s_2+s_1s_5+3s_2s_5-2s_3s_4-2s_4s_5\nonumber \\
 & \qquad -5s_0+6s_2-3s_3-6s_4+10s_5
\end{align}

{We omit terms with zero coefficients in the formulas and edges with zero weights in the graphs.} 
As QUBO and Ising models are represented as graphs, we use graph theory terminologies such as the degree of a node (i.e., the number of edges connecting to a node) and the diameter (i.e., the largest shortest path over all pairs of nodes).
The QUBO and Ising models have optimal solutions $X=[1, 0, 0, 1, 1, 0]$ with energy $E(X)=-7$
and $S=[+1,-1,-1,+1,+1,-1]$ with Hamiltonian $H(S)=-32$, respectively.
These models are equivalent, because $4E(X)=H(S)+4$ always holds
for all $X$ and $S$ satisfying $s_i=2x_i-1$ for all $i$ ($0\leq i\leq n-1$).

\begin{figure}[H]
\includegraphics[scale=0.7]{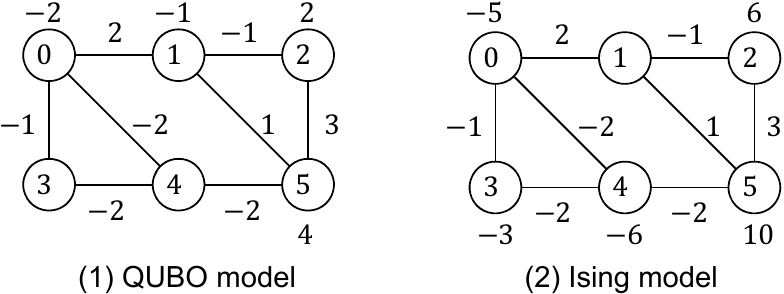}
\caption{Equivalent QUBO and Ising models.}
\label{fig:qubo-ising}
\end{figure}

As many optimization problems can be reduced to QUBO/Ising models~\cite{Lucas14}, there has been a significant effort by researchers to find effective algorithms, hardware, and systems to solve them.
Digital computers and devices can directly operate 0/1 bits, and users/developers can handle them more easily than $-1$/$+1$ qubits.
As a result, QUBO models are more frequently used than Ising models, and QUBO solvers have been developed by many researchers.
On the other hand, quantum annealers based on quantum mechanics can directly operate $-1$/$+1$ qubits in Ising models.
Therefore, for solving QUBO problems on quantum annealers, they are converted to equivalent Ising models.
By this preprocessing step, quantum annealers can be used as QUBO solvers.
However, to maximize the performance of QUBO solvers, we may design Ising models without using QUBO--Ising conversion.

When a QUBO/Ising model for solving a specific permutation-based combinatorial optimization problem is designed,
this involves creating a sub-model that generates any one of all possible permutations as the optimal solution.
This sub-model is referred to as \emph{the QUBO/Ising kernel}.
The main focus of this paper is the design of QUBO/Ising kernels for generating permutations.
For instance, we will present QUBO/Ising kernels
for generating an $n\times n$-matrix that stores one-hot vectors representing a permutation, as shown in Figure~\ref{fig:onehot}.

\subsection{Quantum Annealers and Ising Models}
D-Wave Systems developed a quantum annealer called D-Wave 2000Q~\cite{McGeoch19}, which was based on quantum mechanics.
 It served as a solver for Ising models using a \mbox{2048-node} Chimera graph.
 Subsequently, D-Wave Systems released a more advanced quantum annealer known as D-Wave Advantage~\cite{Advantage},
 which was capable of handling Ising models with a larger 5760-node Pegasus graph~\cite{DWaveAdvantage19}.
 Specifically, the D-Wave Advantage quantum annealer comprises 5760 cells interconnected according to the topology of a 5760-node Pegasus graph.
The cell biases and interaction strengths are programmable, allowing for the acquisition of optimal or approximate solutions to the corresponding Ising models through quantum annealing.

While a quantum annealer is designed to solve Ising models with a specific graph topology, it is possible to solve Ising models with different topologies through a process known as minor embedding.
Minor embedding involves mapping the problem graph (with a different topology) onto the physical graph of the quantum annealer, effectively embedding the problem into the hardware.
This allows the quantum annealer to solve Ising models that may not directly match its native graph topology.
Minor embedding is a technique used to leverage the capabilities of quantum annealers for a broader range of Ising models.
We explain the idea of minor embedding using Figure~\ref{fig:annealer}.
Suppose that we need to solve an Ising problem with the six-node graph in Figure~\ref{fig:qubo-ising}
on a quantum annealer with an eight-cell grid topology, as shown in Figure~\ref{fig:annealer}.
We arrange node $0$ (or qubit $s_0$) in Figure~\ref{fig:qubo-ising} to two cells $0$ (or qubit $s_0$) and $0'$ (or qubit $s_{0'}$) in Figure~\ref{fig:annealer}.
We add a quadratic term $-Ps_{0}s_{0'}$ so that optimal solutions satisfy $s_{0}=s_{0'}$,
where $P$ is a constant number large enough to guarantee it.
By virtue of this embedding, $s_0$ in Figure~\ref{fig:qubo-ising} can be simulated by $s_0$ and $s_{0'}$ combined in Figure~\ref{fig:annealer}.
The reader should have no difficulty in confirming that the Ising model in Figure~\ref{fig:qubo-ising} can be solved
by the quantum annealer in Figure~\ref{fig:annealer}.

\begin{figure}[H]
\includegraphics[scale=0.7]{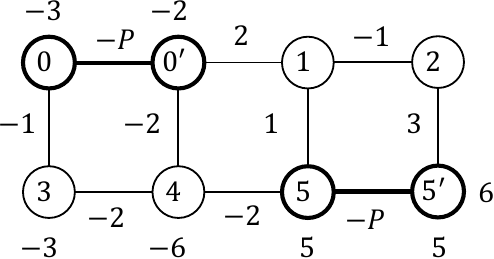}
\caption{Example of minor embedding for the Ising model of Figure~\ref{fig:qubo-ising} to the quantum annealer.}
\label{fig:annealer}
\end{figure}

More than two cells in the quantum annealer may be arranged to simulate a node of the Ising model.
D-Wave Systems calls a set of cells arranged to form a qubit {\em a chain}, and
the value of $P$ the {\em chain strength}~\cite{DWave-chain}.
To obtain an optimal or approximate solution of an Ising model by the D-Wave Advantage quantum annealer,
one needs to find its minor embedding in a 5760-node Pegasus graph.
If an Ising model is a dense graph with a large degree, the chains will be large.
For example, if an Ising model is a complete graph, only 177 nodes
can be embedded in the D-Wave Advantage quantum annealer~\cite{Advantage}.
Hence, the size of an Ising model that can be solved by the D-Wave Advantage quantum annealer is limited.
In particular, the quadratic term count of an Ising model impacts the resource usage of cells.
Thus, it is quite important to minimize the quadratic term count when one designs Ising models.

The D-Wave Advantage quantum annealer~\cite{Advantage} can take real numbers
for the interaction $J_{i,j}$ and bias $h_{i}$ of Ising models.
The ranges of $J_{i,j}$ and $h_{i}$ are limited to $[-1.0,+1.0]$ and $[-4.0,+4.0]$, respectively.
Although they take any real number in these ranges, they are operated as analog values, 
and the effective resolution is limited.
The resolution is only five to six bits, and two values with a difference less than
${1\over 2^6}=0.015$ may not be distinguished due to the flux noise.
Recall that we assume the interaction $J_{i,j}$ and bias $h_{i}$ of Ising models are integers.
When we load such Ising models onto the D-Wave Advantage quantum annealer,
these integer values are automatically reduced to fit in the ranges $[-1.0,+1.0]$ and $[-4.0,+4.0]$.
For example, if $J_{i,j}$ takes integers in the range $[-100,+100]$,
then they will be divided by 100 to fit in the range $[-1.0,+1.0]$, and 
the required resolution $0.01$ is smaller than the effective resolution of the D-Wave Advantage quantum annealer.
Hence, optimal solutions of the Ising model cannot be expected.
Since the resolution is limited, the maximum absolute values of interactions and biases of Ising models
should be as small as possible to obtain optimal solutions with a higher probability by quantum annealers.

QUBO/Ising models with larger diameters may be solvable using a heuristic algorithm in parallel.
They may be split into many disjoint sub-models, and
parallel divide-and-conquer 
techniques can be applied to such models.
Heuristic algorithms to improve the solutions of sub-models can be executed in parallel, and better approximate
solutions can be obtained more quickly.
On the other hand, smaller-diameter models cannot be split into many disjoint sub-models,
and it is difficult to apply parallel heuristic techniques.
Thus, larger-diameter models are preferable to solve the problems in parallel.

\subsection{QUBO/Ising Models for One-Hot/Zero-One-Hot/Domain-Wall Vectors}

\emph{One-hot encoding} has often been used to represent integers in QUBO and Ising models.
A $k$-bit/qubit vector is \emph{a one-hot vector} if it has exactly one $1$/$+1$,
and the remaining $k-1$ bits/qubits take $0$/$-1$.
It represents an integer $i$ ($0\leq i\leq k-1$) if and only if the $i$-th bit/qubit is $1$/$+1$.
For technical reasons, we introduce \emph{the zero-one-hot vector},
which can take a bit/qubit vector with all $0$/$-1$s in addition to one-hot vectors.
Such vectors with all $0$/$-1$s represent a special value $\varphi$ associated with ``undefined'' or ``N/A''. Table~\ref{tab:vectors} shows four-bit/qubit one-hot/zero-one-hot vectors 
with corresponding integers or $\varphi$.

\begin{table}[H]
\tablesize{\footnotesize}
\caption{{One-hot/zero-one-hot/domain-wall} 
 vectors with $k=4$ bits.}
\label{tab:vectors}
		\setlength{\cellWidtha}{\textwidth/9-2\tabcolsep-0.4in}
		\setlength{\cellWidthb}{\textwidth/9-2\tabcolsep+0in}
		\setlength{\cellWidthc}{\textwidth/9-2\tabcolsep+0.4in}
		\setlength{\cellWidthd}{\textwidth/9-2\tabcolsep-0.4in}
		\setlength{\cellWidthe}{\textwidth/9-2\tabcolsep+0in}
		\setlength{\cellWidthf}{\textwidth/9-2\tabcolsep+0.4in}
		\setlength{\cellWidthg}{\textwidth/9-2\tabcolsep-0.4in}
		\setlength{\cellWidthh}{\textwidth/9-2\tabcolsep+0in}
		\setlength{\cellWidthi}{\textwidth/9-2\tabcolsep+0.4in}
		\scalebox{1}[1]{\begin{tabularx}{\textwidth}{>{\centering\arraybackslash}m{\cellWidtha}>{\centering\arraybackslash}m{\cellWidthb}>{\centering\arraybackslash}m{\cellWidthc}>{\centering\arraybackslash}m{\cellWidthd}>{\centering\arraybackslash}m{\cellWidthe}>{\centering\arraybackslash}m{\cellWidthf}>{\centering\arraybackslash}m{\cellWidthg}>{\centering\arraybackslash}m{\cellWidthh}>{\centering\arraybackslash}m{\cellWidthi}}
\toprule
\multicolumn{3}{c}{\textbf{One-Hot}} & \multicolumn{3}{c}{\textbf{Zero-One-Hot}} &\multicolumn{3}{c}{\textbf{Domain-Wall}} \\
  &  \textbf{Bits} & \textbf{Qubits} & &\textbf{Bits} & \textbf{Qubits} & & \textbf{Bits} & \textbf{Qubits} \\
\midrule
0          & $[1,0,0,0]$ & $[+1,-1,-1,-1]$ & 0 &  $[1,0,0,0]$ & $[+1,-1,-1,-1]$ & 0 &$[0,0,0,0]$ & $[-1,-1,-1,-1]$ \\
1          & $[0,1,0,0]$ & $[-1,+1,-1,-1]$ & 1 & $[0,1,0,0]$ & $[-1,+1,-1,-1]$ & 1 &$[1,0,0,0]$ & $[+1,-1,-1,-1]$\\
2          & $[0,0,1,0]$ & $[-1,-1,+1,-1]$ & 2 & $[0,0,1,0]$ & $[-1,-1,+1,-1]$ & 2 &$[1,1,0,0]$ & $[+1,+1,-1,-1]$\\
3          & $[0,0,0,1]$ & $[-1,-1,-1,+1]$ & 3 & $[0,0,0,1]$ & $[-1,-1,-1,+1]$ & 3 & $[1,1,1,0]$ &$[+1,+1,+1,-1]$ \\
           &                &                   & $\varphi$ & $[0,0,0,0]$ & $[-1,-1,-1,-1]$ & 4 & $[1,1,1,1]$ & $[+1,+1,+1,+1]$\\
\bottomrule
\end{tabularx}}
\end{table}

\emph{Domain-wall encoding}~\cite{Codognet22,Chancellor19,Chen21,Berwald22} has also been used to represent integers.
A $k$-bit/qubit vector is {\em a domain-wall vector} 
 if it consists of consecutive $1$/$+1$s following consecutive $0$/$-1$s.
It represents the integer $i$ ($0\leq i\leq k$) if it contains $i$ consecutive $1$s.
Therefore, it can represent $k+1$ integers ranging from 0 to $k$. Table~\ref{tab:vectors} illustrates four-bit/qubit domain-wall vectors that represent integers from 0 to 4.

This sub-section presents QUBO/Ising models for one-hot/zero-one-hot/domain-wall vectors,
aimed at helping the reader understand the fundamental concepts of using these models to represent numbers.
These models are designed to achieve their optimal value if and only if the vectors $X$/$S$ are one-hot/zero-one-hot/domain-wall vectors.

We will first design a QUBO model with $k$-bit variable $X=(x_i)$ ($0\leq i\leq k-1$)
that takes the minimum value of 0 if and only if it stores a $k$-bit one-hot vector. 
It is clear that $X$ is a one-hot vector if and only if the sum of all bits is 1.
Based on this property, we can design a QUBO model with an energy function $E_1(X)$ as follows:
\begin{align}
E_1(X) &= \left(1-\sum\limits_{i=0}^{k-1}x_i\right)^2 = \textcolor{black}{2\sum\limits_{i=0}^{k-2}\sum\limits_{j=i+1}^{k-1}x_ix_j}-\sum\limits_{i=0}^{k-1}x_i+1.
       \label{eq:QUBO-1hot}
\end{align}

It is evident that $E_1(X)$ takes the minimum value of 0 if and only if $X$ is a one-hot vector.
Additionally, we can design a QUBO model that takes the minimum value if and only if it stores a $k$-bit zero-one-hot vector.
A $k$-bit vector $X$ is a zero-one-hot vector if and only if the sum of all bits is either 0 or 1.
Thus, the energy $E_{01}(X)$ defined below takes the minimum value of 0 if $X$ is a zero-one-hot vector.
\begin{align}
E_{01}(X) &= {1\over 2}\sum\limits_{i=0}^{k-1}x_i\left(1-\sum\limits_{i=0}^{k-1}x_i\right) = \textcolor{black}{\sum\limits_{i=0}^{k-2}\sum\limits_{j=i+1}^{k-1}x_ix_j}.
       \label{eq:QUBO-01hot}
\end{align}

QUBO models with $E_{01}(X)$ for zero-one-hot vectors have no linear terms, making them simpler than those with $E_1(X)$ for one-hot vectors.
{Since all quadratic terms in the expanded summation are 2, $E_{01}(X)$ has a coefficient of $1\over 2$.}
Our objective is to minimize the integer coefficients of both linear and quadratic terms.
Consequently, we will employ such coefficients whenever possible throughout the remainder of this paper.

Similarly, we can design Ising models with $H_1(S)$ and $H_{01}(S)$
for a $k$-qubit variable $S=(s_i)$ ($0\leq i\leq k-1$) storing one-hot and zero-one-hot vectors, respectively.
A $k$-qubit variable $S$ stores a one-hot vector if and only if
$\sum\limits_{i=0}^{k-1}s_i=1\cdot (+1)+(k-1)\cdot(-1)=-(k-2)$.
Thus, the following $H_1(S)$ takes the minimum value of 0 if and only if $S$ is a one-hot vector:
\begin{align}
H_1(S) &= {1\over 2}\left((k-2)+\sum\limits_{i=0}^{k-1}s_i\right)^2   = \textcolor{black}{\sum\limits_{i=0}^{k-2}\sum\limits_{j=i+1}^{k-1}s_is_j}+ (k-2)\sum\limits_{i=0}^{k-1}s_i +{1\over 2}k^2-{3\over 2}k+2.
\label{eq:Ising-1hot}
\end{align}

Clearly, $H_1(S)$ takes the minimum value of 0 if and only if $S$ is a one-hot vector.
Additionally, $S$ stores a zero-one-hot vector if and only if it has zero or one $+1$.
Thus, the following Ising model with $H_{01}(S)$ takes the minimum value of 0 if and only if $S$ is a zero-one-hot vector, which means that the sum of all qubits is $-k$ or $-(k-2)$:
\begin{align}
H_{01}(S) &= {1\over 2}\left(k+\sum\limits_{i=0}^{k-1}s_i\right)\left((k-2)+\sum\limits_{i=0}^{k-1}s_i\right)  = \textcolor{black}{\sum\limits_{i=0}^{k-2}\sum\limits_{j=i+1}^{k-1}s_is_j+ (k-1)\sum\limits_{i=0}^{k-1}s_i} +{1\over 2}k^2-{1\over 2}k.
\label{eq:Ising-01hot}
\end{align}

We will now design a QUBO model for a $k$-bit variable $X = (x_i)$ ($0 \leq i \leq k-1$) that stores a domain-wall vector.
For technical reasons, we assume the existence of fixed \emph{guard bits} $x_{-1} = 1$ and $x_k = 0$ for $X$.
With these guard bits, if $X$ stores a domain-wall vector, then $x_{i-1} - x_i \neq 0$ ($0 \leq i \leq k$) holds for exactly one $i$.
Otherwise, it holds for more than two $i$.
Based on this fact, we have the following QUBO model:
\begin{align}
E_d(X) &= {1\over 2}\sum_{i=0}^{n-1} (x_{i-1}-x_{i})^2 = - \sum\limits_{i=1}^{k-1}x_{i-1}x_i+\sum\limits_{i=1}^{k-1}x_i +{1\over 2}.
\label{eq:QUBO-domain}
\end{align}

The energy function $E_d(X)$ takes the minimum value of $1\over 2$ if and only if $(x_{i-1}-x_{i})^2=1$ for exactly one $i$ and $X$ stores a domain-wall vector.
Similarly, we can design an Ising model for a $k$-qubit variable $S=(s_i)$ ($0\leq i\leq k-1$) that stores a domain-wall vector using the same approach.
We also assume the existence of fixed \emph{guard qubits} $s_{-1}=+1$ and $s_{k}=-1$.
The following Ising model with $H_d(S)$ takes the minimum value of 2
if and only if $S$ stores a domain-wall vector and $(s_{i-1}-s_{i})^2=4$ for exactly one $i$:
\begin{align}
H_d(S) &= {1\over 2} \sum_{i=0}^{k} (s_{i-1}-s_i)^2
 = -\sum\limits_{i=1}^{k-1}s_{i-1}s_i - s_0 +s_{k-1} + {(k+1)} 
 \label{eq:Ising-domain}
\end{align}

Table~\ref{tab:1d-stat} summarizes the features of QUBO/Ising models with $k$-bit/qubit one-hot/zero-one-hot/domain-wall vectors.
We utilized SymPy \cite{Sympy}, a Python library for symbolic mathematics, to expand the mathematical formulas and derive the features of the QUBO/Ising models presented throughout this paper.
We observe that QUBO/Ising models for one-hot vectors are fully connected and consist of ${1\over 2}k^2-{1\over 2}k$ quadratic terms.
In contrast, models for domain-wall vectors have only $k-1$ quadratic terms.
Furthermore, the linear term coefficients of Ising models for one-hot and zero-one-hot vectors are $k-2$ and $k-1$, respectively.
On the other hand, those for domain-wall vectors have only two linear terms with coefficients $-1$ and $+1$, respectively.

\begin{table}[H]
\tablesize{\small}
\caption{QUBO/Ising models for $k$-bit/qubit one-hot/domain-wall vectors.} 
\label{tab:1d-stat}
		\setlength{\cellWidtha}{\textwidth/4-2\tabcolsep+0.9in}
		\setlength{\cellWidthb}{\textwidth/4-2\tabcolsep-0.3in}
		\setlength{\cellWidthc}{\textwidth/4-2\tabcolsep-0.3in}
		\setlength{\cellWidthd}{\textwidth/4-2\tabcolsep-0.3in}
		\scalebox{1}[1]{\begin{tabularx}{\textwidth}{>{\raggedright\arraybackslash}m{\cellWidtha}>{\centering\arraybackslash}m{\cellWidthb}>{\centering\arraybackslash}m{\cellWidthc}>{\centering\arraybackslash}m{\cellWidthd}}
		\toprule
\textbf{Encoding Type}  & \textbf{One-Hot} & \textbf{Zero-One-Hot} & \textbf{Domain-Wall} \\
 \midrule
 QUBO models &\\
Quadratic formula & $E_1(X)$& $E_{01}(X)$ & $E_{d}(X)$ \\
Linear term count & $k$ & $0$ &  $k-1$\\
Linear term coefficients & $-1$ & - & $+1$\\
Quadratic term count & ${1\over 2}k^2-{1\over 2}k$ & ${1\over 2}k^2-{1\over 2}k$ & {$k-1$} 
\\
Quadratic term coefficients & $+2$ & $+1$ & $-1$ \\
Diameter & $1$ &$1$ & $k-1$ \\
Offset & $1$ & $0$ & $1\over 2$\\
Optimal energy &$0$ & $0$ & $1\over 2$\\
 \midrule
Ising models &\\
Quadratic formula & $H_1(S)$ & $H_{01}(S)$ & $H_{d}(S)$ \\
Linear term count & $k$ & $k$ &  $2$\\
Linear term coefficients & $k-2$ & $k-1$ & {$-1$, $+1$}\\ 
Quadratic term count & ${1\over 2}k^2-{1\over 2}k$ & ${1\over 2}k^2-{1\over 2}k$ & {$k-1$}\\ 
Quadratic term coefficients & $+1$ & $+1$ & $-1$ \\
Diameter & $1$ &$1$ & $k-1$ \\
Offset & ${1\over 2}k^2-{3\over 2}k+2$ & ${1\over 2}k^2-{1\over 2}k$ & $k+1$\\
Optimal Hamiltonian &$0$ & $0$ & $2$\\
\bottomrule
\end{tabularx}}
\end{table}

\section{QUBO/Ising Model Kernels for Generating Permutation Involving a Cubic Number of Quadratic Terms}
\label{sec:conventional}
A permutation of $n$ numbers can be defined by a bijection $\pi:\{0,1,\ldots, n-1\}\rightarrow \{0,1,\ldots, n-1\}$, where the list $[\pi(0), \pi(1), \ldots, \pi(n-1)]$ represents one of the $n!$ possible permutations.
This section initially describes a conventional permutation encoding method that employs one-hot vectors.
This approach is widely used for solving permutation-based combinatorial optimization problems, not just in QUBO/Ising models but also in mixed-integer programming.
We then explain a permutation encoding method using domain-wall vectors, which was introduced in~\cite{Codognet22} and reduces the number of quadratic terms by half.

\subsection{Conventional Permutation Encoding by One-Hot Vectors}
A permutation $\pi$ is commonly represented by an $n\times n$ matrix of variables, where each row $i$ ($0 \leq i \leq n-1$) stores $\pi(i)$ as a one-hot vector.
This representation ensures that each row of the matrix is a one-hot vector.
We will design $n^2$-bit QUBO/Ising kernels with $X=(x_{i,j})$/$S=(s_{i,j})$ ($0\leq i,j\leq n-1$),
which can generate permutations as the optimal solutions.
For this purpose, we apply $E_1(X)$/$H_1(S)$ to all rows and columns
to guarantee that they are one-hot vectors as follows:
\begin{align}
E_1^{nn}(X) &= 
 {1\over 2}\sum_{i=0}^{n-1} \left(1-\sum\limits_{j=0}^{n-1}x_{i,j}\right)^2+
 {1\over 2}\sum_{j=0}^{n-1} \left(1-\sum\limits_{i=0}^{n-1}x_{i,j}\right)^2
       \label{eq:QUBO-1hot-nn}
\end{align}
\begin{align}
H_1^{nn}(S) &= 
 {1\over 2}\sum_{i=0}^{n-1}\left((n-2)+\sum\limits_{j=0}^{n-1}s_{i,j}\right)^2+
 {1\over 2}\sum_{j=0}^{n-1}\left((n-2)+\sum\limits_{i=0}^{n-1}s_{i,j}\right)^2
              \label{eq:Ising-1hot-nn}
\end{align}

We can obtain the quadratic formulas for QUBO and Ising models by expanding $E_1^{nn}(X)$ and $H_1^{nn}(S)$, respectively.
Both formulas take the minimum value of 0 if and only if every row and every column is a one-hot vector.
The features of QUBO/Ising models obtained by expanding $E_1^{nn}(X)$/$H_1^{nn}(S)$ can be found in Table~\ref{tab:permutation-stat}.
It is important to note that both models have $n^3-n^2$ quadratic terms, and the Ising model also includes linear terms with a coefficient of $2n-4$.

\begin{table}[H]
\caption{QUBO/Ising kernels for generating a permutation of $n$ numbers.} 
\label{tab:permutation-stat}
	\begin{adjustwidth}{-\extralength}{0cm} 
		\setlength{\cellWidtha}{\fulllength/5-2\tabcolsep+0.3in}
		\setlength{\cellWidthb}{\fulllength/5-2\tabcolsep-0.2in}
		\setlength{\cellWidthc}{\fulllength/5-2\tabcolsep+0.2in} 
		\setlength{\cellWidthd}{\fulllength/5-2\tabcolsep-0.1in}
				\setlength{\cellWidthd}{\fulllength/5-2\tabcolsep-0.2in}
		\scalebox{1}[1]{\begin{tabularx}{\fulllength}{>{\raggedright\arraybackslash}m{\cellWidtha}>{\centering\arraybackslash}m{\cellWidthb}>{\centering\arraybackslash}m{\cellWidthc}>{\centering\arraybackslash}m{\cellWidthd}>{\centering\arraybackslash}m{\cellWidthd}}
			\toprule
                      & & \multicolumn{3}{c}{\textbf{Domain-Wall}} \\
\textbf{Encoding Type}  & \textbf{One-Hot} & \textbf{All-Different} & \textbf{Dual-Matrix} & \textbf{Extended} \\
 \midrule
Bit/qubit count & $n^2$ & $n^2-n$ & $2n^2-2n$ & $3n^2-2n$  \\
 \midrule
 QUBO models &&&&\\
Quadratic formula & $E_1^{nn}(X)$& $E_a^{nn}(X)$ & $E_d^{nn}(A,B)$ & {$E_e^{nn}(A,B,X)$}\\ 
Linear term count &  $n^2$ & $n^2-2n$ & $2n^2-4n$ & $3n^2-6n+2$ \\
Linear term coefficients & $-1$ & {$-(2n-3), -(2n-6), \ldots, -2$} & $+2$ & $-1,+1,+2$\\ 
Quadratic term count &  {$n^3-n^2$} 
 & {${1\over 2}n^3-{3\over 2}n$}& {$6n^2-12n+4$} & {$6n^2-8n$}\\  
Quadratic term coefficients & $+1$ & $-1, +2$ & $-2,-1,+1$ & $-2,-1,+1$\\
Diameter &  $2$ & $n-1$ & $2n-3$ & $2n$ \\
Offset & $n$  & ${1\over 3}n^3-{1\over 2}n^2-2n$ & $2n-1$ & $2n$\\
Optimal energy & $0$ & ${1\over 2}n$ & $n$ & $n$ \\
 \midrule
Ising models &&&&\\
Quadratic formula &  $H_1^{nn}(S)$ &$H_a^{nn}(S)$ & $H_d^{nn}(A,B)$ & {$H_e^{nn}(A,B,S)$}\\ 
Linear term count & $n^2$ & $n^2 + n(n\bmod 2) - 2n$  &$4n$ & $n^2+4n-4$ \\
Linear term coefficients & {$2n-4$} & {$-(n-1), -(n-4),\ldots, +(n-1)$} & {$-2,+2$} & {$-2,+1,+2$} \\ 
Quadratic term count & {$n^3-n^2$} & {${1\over 2}n^3-{3\over 2}n$} & {$6n^2-12n+4$} & {$6n^2-8n$} \\ 
Quadratic term coefficients &  $+1$ & $-1,+1$& $-2,-1,+1$ & $-2,-1,+1$ \\
Diameter &  $2$ & $n-1$ & $2n-3$ & $2n$ \\
Offset & $n^3-3n^2+4n$ & ${1\over 6}n^3 + n^2- {1\over 6}n$&$4n^2-4$ &$6n^2-4n$ \\
Optimal Hamiltonian & $0$ & $2n$ & $4n$ & $4n$\\
 \bottomrule
\end{tabularx}}
\end{adjustwidth}
\end{table}

\subsection{Permutation Encoding by All-Different Domain-Wall Encoding}
We will now introduce \emph{the all-different domain-wall technique}, which was presented in~\cite{Codognet22} and can be used to generate permutations
with domain-wall vectors.
This technique utilizes a matrix $X=(x_{i,j})$ of size $n\times (n-1)$, where each row $i$ ($0\leq i\leq n-1$) stores the domain-wall vector representing $\pi(i)$ for a permutation $\pi$. 
Figure~\ref{fig:onehot} illustrates an example with $n=4$.
It is important to note that a matrix $X$, where each row contains a domain-wall vector, represents a permutation if and only if all the domain-wall vectors are distinct.
The all-different domain-wall technique leverages this property to design a QUBO kernel for generating a permutation, as follows:
\begin{align}
E_a^{nn}(X) &= 
 {1\over 2} \sum_{i=0}^{n-1}\sum_{j=0}^{n-1} (x_{i,j-1}-x_{i,j})^2+
\sum_{j=0}^{n-2}\left((n-j-1)-\sum\limits_{i=0}^{n-1}x_{i,j}\right)^2
       \label{eq:QUBO-all-different}
\end{align}

Here, fixed guard bits $x_{i,-1}=1$ and $x_{i,n-1}=0$ for all $i$ ($0\leq i\leq n-1$) are used.
The first summation term takes the minimum value of ${1\over 2}n$ if and only if every row is a domain-wall vector.
Also, the the second summation term takes the minimum value of 0 if and only if the number of 1s in each column $j$ ($0\leq j\leq n-1$) is $n-j-1$.
If this is the case, all rows of $X$ store distinct domain-wall vectors.
Therefore, $X$ is a permutation if and only if it is an optimal solution that satisfies $E_a^{nn}(X)={1\over 2}n$. 

We can apply the same technique for designing an Ising kernel
with an $n\times (n-1)$ matrix $S=(s_{i,j})$ that generates a permutation as a domain-wall vector.
If all rows in $S$ store distinct domain-wall  vectors, then the number of $+1$s in each column $j$ ($0\leq j\leq n-1$) is $(n-j-1)$,
and the sum of all elements in it is $(n-j-1)\dot(+1)+(j+1)\cdot(-1)=n-2j-2$.
Based on this fact, we can design a desired Ising model as follows:
\begin{align}
H_a^{nn}(X) &= 
 {1\over 2} \sum_{i=0}^{n-1}\sum_{j=0}^{n-1} (s_{i,j-1}-s_{i,j})^2+
 {1\over 2} \sum_{j=0}^{n-2}\left((n-2j-2)-\sum\limits_{i=0}^{n-1}s_{i,j}\right)^2
       \label{eq:Ising-all-different}
\end{align}

Here, $s_{i,-1}=+1$ and $s_{i,n-1}=-1$ for all $i$ ($0\leq i\leq n-1$).
The first summation term takes the minimum value of ${1\over 2}\cdot 4n=2n$ if and only if every row is a domain-wall vector.
Also, the second summation term takes 0 if and only if the number of $+1$s in each column $j$ ($0\leq j\leq n-1$) is $n-j-1$.
If this is the case, all rows of $S$ store distinct domain-wall vectors.
Therefore, $S$ stores domain-wall vectors representing a permutation if and only if it is an optimal solution that satisfies $H_a^{nn}(X)=2n$.

Table~\ref{tab:permutation-stat} presents the features of QUBO and Ising models derived from the expansion of mathematical expressions
$E_a^{nn}(X)$ and $H_a^{nn}(S)$.
The number of quadratic terms in these models is halved compared to conventional one-hot encoding.
However, these models include linear terms with significantly large absolute values.
The maximum absolute values for the QUBO and Ising models are $2n-3$ and $n-1$, respectively.

\section{QUBO/Ising Kernels with a Quadratic Number of Quadratic Terms Based on Our Dual-Matrix Domain-Wall Technique}
\label{sec:dual-matrix}
Our new permutation encoding technique utilizes the inverse permutation. 
For a permutation $\pi$ of $n$ numbers, let $\pi^{-1}$ denote the inverse, such that $\pi^{-1}(\pi(i))=i$ for all $i$ ($0\leq i\leq n-1$).
We refer to the permutations $[\pi(0), \pi(1), \ldots, \pi(n-1)]$ and $[\pi^{-1}(0), \pi^{-1}(1), \ldots, \pi^{-1}(n-1)]$ as {\emph{dual permutations}}. 

Recall that in conventional one-hot encoding, each row $i$ ($0 \leq i \leq n-1$) of an $n \times n$ matrix stores $\pi(i)$ as a one-hot vector, as illustrated in Figure~\ref{fig:onehot}. 
It is easy to confirm that each column $j$ ($0 \leq j \leq n-1$) also stores a one-hot vector representing the inverse $\pi^{-1}(j)$. 
Therefore, the row one-hot vectors and column one-hot vectors represent dual permutations.
Figure~\ref{fig:onehot} shows an example of a $4 \times 4$ matrix storing the row permutation $[1, 3, 2, 0]$ and the column permutation $[3, 0, 2, 1]$, which are dual permutations.
Our dual-matrix domain-wall technique was inspired by this fact.
It employs two matrices: matrix $A=(a_{i,j})$ ($0\leq i\leq n-1$ and $0\leq j \leq n-2$) of size $n\times (n-1)$ and matrix $B=(b_{i,j})$
 (${0}\leq i\leq n-2$ and $0\leq j \leq n-1$) of size $(n-1)\times n$. 
These matrices are used to store dual permutations, where the rows of $A$ represent a permutation and the columns of $B$ represent its inverse permutation.
{For a clearer understanding, the reader should refer to Figure~\ref{fig:dual-matrix}, which provides an example for $n=4$ using a $4\times 3$ matrix $A$ and a $3\times 4$ matrix $B$.}

\begin{figure}[H]
	\begin{adjustwidth}{-\extralength}{0cm}
\includegraphics[scale=0.68]{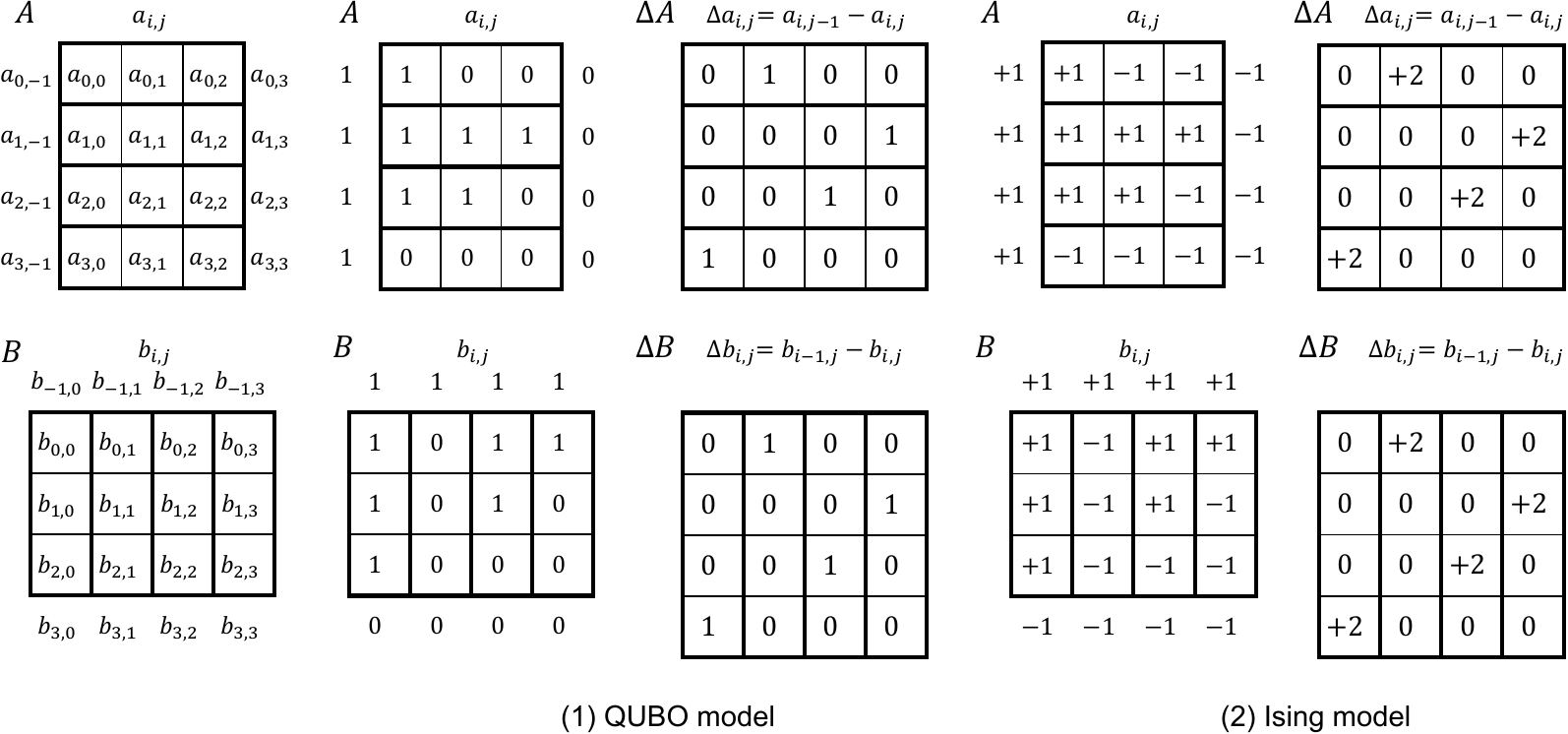}
	\end{adjustwidth}
\caption{The dual-matrix domain-wall technique for representing permutations of $n=4$ numbers on an Ising model.}
\label{fig:dual-matrix}
\end{figure}

For domain-wall encoding, fixed guard bits/qubits are attached to the leftmost and rightmost columns of matrix $A$, as well as the top and bottom rows of matrix $B$, as illustrated in Figure~\ref{fig:dual-matrix}.
{Specifically, we set $a_{i,-1}=1/+1$ and $a_{i,n-1}=0/-1$ for all $i$ ($0\leq i\leq n-1$), and $b_{-1,j}=1/+1$ and $b_{n-1,j}=0/-1$ for all $j$ ($0\leq i\leq n-1$) for QUBO/Ising models}.  
Our goal is to design QUBO/Ising models that have optimal solutions satisfying the following three conditions:

\textbf{{Row domain-wall condition:}
} Each row of matrix $A$ stores a domain-wall vector.

\textbf{{Column domain-wall condition}}: Each column of matrix $B$ stores a domain-wall vector.

\textbf{{Dual-permutation condition}}: The row permutation of matrix $A$ and the column permutation of matrix $B$ are dual to each other.

An optimal solution of the designed QUBO/Ising models corresponds to any one of the possible $n!$ permutations.

We begin by introducing the dual-matrix domain-wall technique for QUBO models.
We define matrices $\Delta A=(\Delta a_{i,j})$ and $\Delta B=(\Delta b_{i,j})$ of size $n\times n$ ($0\leq i,j\leq n-1$) by computing the differences in each row of matrix $A$ and each column of matrix $B$ as follows:
\begin{alignat}{2}
\Delta a_{i,j} &= a_{i,j-1}-a_{i,j} &\mbox{ and } \Delta b_{i,j} &= b_{i-1,j}-b_{i,j}.
\label{eq:delta}
\end{alignat}

Figure~\ref{fig:dual-matrix} illustrates examples of $\Delta A$ and $\Delta B$ with $n=4$.
The QUBO kernel for generating a permutation using the dual-matrix domain-wall technique is defined by the following expression, denoted as $E_d^{nn}(A,B)$:
\begin{align}
E_d^{nn}(A,B) &={1\over 2}\sum_{i=0}^{n-1}\sum_{j=0}^{n-1} (\Delta a_{i,j})^2+{1\over 2}\sum_{i=0}^{n-1}\sum_{j=0}^{n-1}(\Delta b_{i,j})^2
 +{1\over 2}\sum_{i=0}^{n-1}\sum_{j=0}^{n-1} (\Delta a_{i,j}-\Delta b_{i,j})^2
\label{eq:Ed_nn}
\end{align}

The first summation term takes the minimum value of ${1\over 2}n$ if and only if the row domain-wall condition is satisfied.
Similarly, the second summation term takes the minimum value of ${1\over 2}n$ if and only if the column domain-wall condition is satisfied.
Lastly, the third summation term takes the minimum value of $0$ if and only if $\Delta a_{i,j}=\Delta b_{i,j}$ for all $i$ and $j$.
In essence, this condition signifies the satisfaction of the dual-permutation condition.
It is important to note that these three summation terms can all reach their minimum values simultaneously.
Consequently, the energy function $E_d^{nn}(A,B)$ reaches its minimum value of $n$ if and only if matrices $A$ and $B$ store dual permutations, satisfying all three conditions: the row domain-wall condition, the column domain-wall condition, and the dual-permutation condition.

For the Ising model, we can use the exact same mathematical expression as the QUBO model, resulting in the Hamiltonian function $H_d^{nn}(A,B)$:
\begin{align}
H_d^{nn}(A,B) &={1\over 2}\sum_{i=0}^{n-1}\sum_{j=0}^{n-1} (\Delta a_{i,j})^2+{1\over 2}\sum_{i=0}^{n-1}\sum_{j=0}^{n-1}(\Delta b_{i,j})^2
 +{1\over 2}\sum_{i=0}^{n-1}\sum_{j=0}^{n-1} (\Delta a_{i,j}-\Delta b_{i,j})^2
\label{eq:Hd_nn}
\end{align}

Both the first and second summation terms take the minimum value of $2n$ if and only if both the row and column domain-wall conditions are satisfied.
The third summation term takes the minimum value of $0$ if and only if the dual-permutation condition is satisfied.
Since these three summation terms can take the minimum values at the same time, the Hamiltonian function $H_d^{nn}(A,B)$ takes the minimum value of $4n$ if and only if matrices $A$ and $B$ store dual permutations.
It is worth noting that although $E_d^{nn}(A,B)$ and $H_d^{nn}(A,B)$ have the same mathematical expressions, their expansions differ.
The energy function $E_d^{nn}(A,B)$ represents the QUBO model and takes the minimum value of $n$ when all three conditions (row domain-wall, column domain-wall, and dual-permutation) are met.
On the other hand, the Hamiltonian function $H_d^{nn}(A,B)$ corresponds to the Ising model and achieves the minimum value of $4n$ when matrices $A$ and $B$ satisfy the dual-permutation condition along with the row and column domain-wall conditions.
These QUBO/Ising models can be used as kernels for solving permutation-based combinational optimization problems as QUBO/Ising models.

The features of QUBO/Ising models obtained using our dual-matrix domain-wall encoding technique are presented in Table~\ref{tab:permutation-stat}.
It is evident from the table that these models consist of only $6n^2-12n+4$ quadratic terms.
In contrast, models designed using the conventional one-hot encoding approach contain $n^3-n^2$ quadratic terms, while models created using the all-different domain-wall encoding method have ${1\over 2}n^3-{3\over 2}n$ quadratic terms.
Thus, our dual-matrix domain-wall technique significantly reduces the number of quadratic terms in the models.

Moreover, our encoding technique also leads to a reduction in the linear term coefficients in the Ising models.
In Ising models obtained using conventional one-hot encoding and all-different domain-wall encoding, the linear terms have coefficients of $2n-4$ and $n-1$, respectively.
However, Ising models obtained through dual-matrix domain-wall encoding feature linear terms with a coefficient of $2$.

\section{QUBO/Ising Kernels for Partial Permutations}
\label{sec:partial}
For two positive integers satisfying $m \leq n$, a \emph{partial permutation} refers to a sequence of $m$ numbers selected from a set ${0, 1, \ldots, n-1}$ of $n$ numbers without repetition.
A permutation is a special case of a partial permutation where $m = n$.
In this section, we primarily focus on the case where $m < n$.
However, we can also consider a permutation of $n$ numbers by substituting $m$ with $n$.
A partial permutation can be represented by an injection $\pi: \{0,1,\ldots, m-1\} \rightarrow \{0, 1, \ldots, n-1\}$, where $\pi(i)$ denotes the $i$-th element in the partial permutation $[\pi(0),\pi(1),\ldots,\pi(m-1)]$.
The total number of possible partial permutations is given by $n!/(n-m)!$.
\emph{An inverse} of a partial permutation $\pi$ is a surjection $p:\{0,1,\ldots,n-1\}\rightarrow\{0,1,\ldots,m-1\}$ that satisfies $p(\pi(i))=i$ for all $i$ ($0\leq i\leq m-1$). 
It should be noted that $p(j)$ can take any value when $\pi^{-1}(j)$ is undefined, meaning that there is no $i$ satisfying $\pi(i)=j$.
We refer to such $\pi$ and $p$ as \emph{dual permutations}.

This section first demonstrates partial-permutation encoding using one-hot vectors.
Subsequently, we illustrate the application of the dual-matrix domain-wall technique to generate partial permutations.
It is important to note that the all-different domain-wall technique cannot be applied in this case.
Unlike permutation generation, the column-wise sums of the domain-wall encoding for a partial permutation are not fixed.

\subsection{Partial-Permutation Encoding by One-Hot Vectors}
We can construct QUBO/Ising kernels to generate partial permutations using zero-one-hot vectors.
For this purpose, we utilize an $m\times n$ matrix ($m<n$) of bits/qubits.
Each row $i$ ($0\leq i\leq m-1$) in the matrix represents $\pi(i)$ as a one-hot vector, while each column $j$ ($0\leq j\leq n-1$) represents $\pi^{-1}(j)$ as a zero-one-hot vector.
Here, we write $\pi^{-1}(j)=\varphi$ if there is no $i$ satisfying $\pi(i)=j$.
Figure~\ref{fig:onehot-phi} illustrates an example of a $3\times 5$ matrix that represents a partial permutation.
In this example, the rows correspond to the partial permutation $[\pi(0),\pi(1),\pi(2)]=[3, 0, 1]$ selected from the set $\{0, 1, 2, 3, 4\}$, and the columns represent the inverse $[\pi^{-1}(0),\pi^{-1}(1),\pi^{-1}(2)$, $\pi^{-1}(3)$ ,$\pi^{-1}(4)]=[1, 2, \varphi, 0, \varphi]$.

\begin{figure}[H]
\includegraphics[scale=0.7]{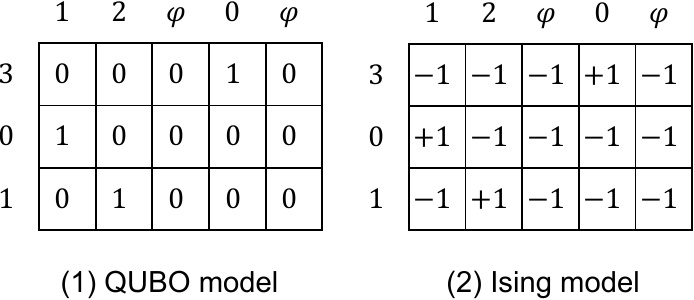}
\caption{$3\times 5$ matrix representing a partial permutation selecting 3 from 5  using one-hot and zero-one-hot vectors.}
\label{fig:onehot-phi}
\end{figure}

We can design the QUBO/Ising kernels $E_1^{mn}(X)$/$H_1^{mn}(S)$ generating a partial permutation as follows:
\begin{align}
E_1^{mn}(X) &= 
 {1\over 2}\sum_{i=0}^{m-1} \left(1-\sum\limits_{j=0}^{n-1}x_{i,j}\right)^2+
 {1\over 2}\sum_{j=0}^{n-1} \sum\limits_{i=0}^{n-1}x_{i,j}\left(1-\sum\limits_{i=0}^{m-1}x_{i,j}\right)
       \label{eq:QUBO-1hot-mn}\\
H_1^{mn}(S) &= 
 {1\over 2}\sum_{i=0}^{m-1}\left((n-2)+\sum\limits_{j=0}^{n-1}s_{i,j}\right)^2+
 {1\over 2}\sum_{j=0}^{n-1}\left(m+\sum\limits_{i=0}^{m-1}s_{i,j}\right)\left((m-2)+\sum\limits_{i=0}^{m-1}s_{i,j}\right)
              \label{eq:Ising-1hot-mn}
\end{align} 

The first summation term in each expression takes the minimum value of 0 when every row is a one-hot vector.
Similarly, the second summation term is evaluated as 0 when every column is a zero-one-hot vector.
Consequently, the optimal solutions of $E_1^{mn}(X)$/$H_1^{mn}(S)$ represent dual permutations.

We should note that previous works such as~\cite{Calude17,Yoshimura21} have already presented QUBO models for partial permutation using zero-one-hot vectors.
However, the mathematical expressions used in these works differ.
In~\cite{Calude17}, a slack variable $y_i$ ($0\leq i\leq n-1$) was introduced for each column $i$, and the expression $\sum\limits_{j=0}^{n-1} \left(1-y_i-\sum\limits_{i=0}^{m-1}x_{i,j}\right)^2$ was utilized to ensure that each column represented a zero-one-hot vector.
Similarly, in~\cite{Yoshimura21}, the expression $\sum\limits_{j=0}^{n-1}\left({1\over 2}-\sum\limits_{i=0}^{m-1}x_{i,j}\right)^2$ was employed for the same purpose.
In contrast, our proposed mathematical expression $\sum\limits_{i=0}^{n-1}x_{i,j}\left(1-\sum\limits_{i=0}^{m-1}x_{i,j}\right)$ for zero-one-hot vectors is more intuitive and simpler as it does not yield any linear terms.

\subsection{Partial-Permutation Encoding by Our Dual-Matrix Domain-Wall Encoding Technique}
We utilize two matrices, namely $A=(a_{i,j})$ ($0\leq i\leq m-1$ and $0\leq j\leq n-2$) of size $m\times(n-1)$ with bit/qubit variables and $B=(b_{i,j})$ ($0\leq i\leq m-1$ and $0\leq j\leq n-1$) of size $(m-1)\times n$.
These matrices are employed for our dual-matrix domain-wall technique and represent partial permutations.
For an illustration, refer to Figure~\ref{fig:domain-partial}, which shows example matrices $A$ and $B$ with dimensions $m=3$ and $n=5$, respectively.
Similarly to dual-matrix domain-wall encoding for permutations, we include fixed guard bits/qubits in the leftmost and rightmost columns of $A$, as well as the top and bottom rows of $B$, as depicted in Figure~\ref{fig:domain-partial}.
By applying Equation~(\ref{eq:delta}), we can derive two matrices $\Delta A=(\Delta a_{i,j})$ and $\Delta B=(\Delta b_{i,j})$ ($0\leq i\leq m-1$, $0\leq j\leq n-1$).

\begin{figure}[H]
	\begin{adjustwidth}{-\extralength}{0cm}
\includegraphics[scale=0.62]{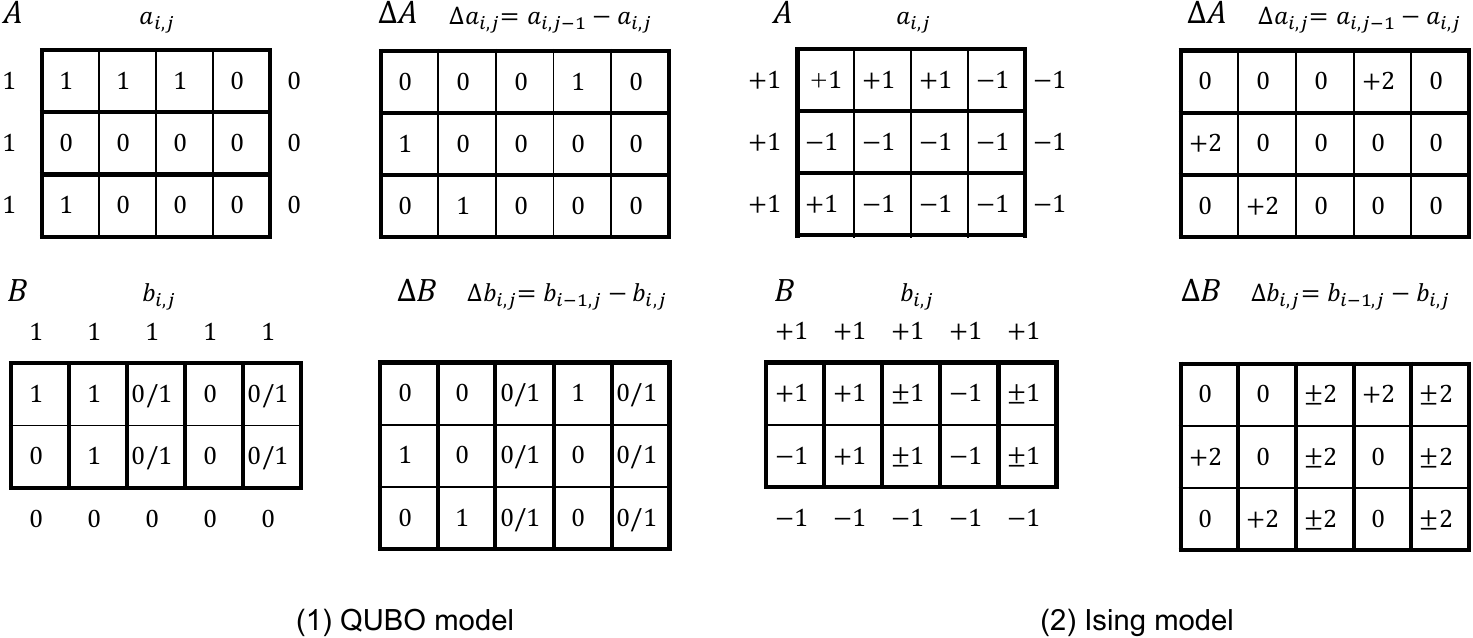}
	\end{adjustwidth}
\caption{Dual-matrix domain-wall encoding presenting a partial permutation selecting $3$ from $5$ using domain-wall vectors.}
\label{fig:domain-partial}
\end{figure}

To satisfy the three conditions, we make modifications to the mathematical expressions $E_d^{nn}(A,B)$ and $H_d^{nn}(A,B)$ for the QUBO and Ising kernels that generate permutations.
These modifications ensure that the row domain-wall condition, column domain-wall condition, and dual-permutation condition are met by the two matrices $A$ and $B$ with sizes $m\times(n-1)$ and $(m-1)\times n$, respectively.
The adjusted expressions are as follows:
\begin{align}
E_d^{mn}(A,B) &={1\over 2}\sum_{i=0}^{m-1}\sum_{j=0}^{n-1} (\Delta a_{i,j})^2+{1\over 2}\sum_{i=0}^{m-1}\sum_{j=0}^{n-1}(\Delta b_{i,j})^2
 +{1\over 2}\sum_{i=0}^{m-1}\sum_{j=0}^{n-1} (\Delta a_{i,j}-\Delta b_{i,j})^2
\label{eq:Ed_mn}\\
H_d^{mn}(A,B) &={1\over 2}\sum_{i=0}^{m-1}\sum_{j=0}^{n-1} (\Delta a_{i,j})^2+{1\over 2}\sum_{i=0}^{m-1}\sum_{j=0}^{n-1}(\Delta b_{i,j})^2
 +{1\over 2}\sum_{i=0}^{m-1}\sum_{j=0}^{n-1} (\Delta a_{i,j}-\Delta b_{i,j})^2
\label{eq:Hd_mn}
\end{align}

We proceed to demonstrate that $E_d^{mn}(A,B)$ attains its minimum value of $n-m$ if and only if
$\Delta A$ represents a partial permutation.
The first summation term takes the minimum value of ${1\over 2}m$
when all rows of matrix $A$ are domain-wall vectors, satisfying the row domain-wall condition. 
Similarly, the second summation terms will reach the minimum values of ${1\over 2}n$ if the column domain-wall condition is satisfied.
Suppose that the row domain-wall condition and the column domain-wall condition are satisfied.
Since the numbers of 1s in {$\Delta A$ and $\Delta B$} are $m$ and $n$, respectively, the third summation term takes the minimum value of ${1\over 2}(n-m)$ if
the dual-permutation condition is satisfied.
Hence, $E_d^{mn}(A,B)$ takes the value ${1\over 2}m+{1\over 2}n+{1\over 2}(n-m)=n$ if the three conditions are satisfied.
Unfortunately, this does not imply that $E_d^{mn}(A,B)$ takes the minimum value of $n$ if and only if the three conditions are satisfied.
The third summation term can be smaller than ${1\over 2}(n-m)$, say 0, if the row or column domain-wall condition is not satisfied.
Hence, we need to prove that $E_d^{mn}(A,B)$ is larger than $n$ if at least one of the three conditions is not satisfied.
More specifically, we prove the following lemma:
\begin{lem}
\label{lem:Ed_nn}
$E_d^{mn}(A,B)$ is larger than $n$ if $A$ and $B$ do not satisfy at least one of the row domain-wall, column domain-wall,
and dual-permutation conditions.
\end{lem}
\begin{proof}
Suppose that the row and column domain-wall conditions are satisfied and the numbers of 1s in $\Delta A$ and $\Delta B$ are $m$ and $n$, respectively.
We can observe that the third summation term of $E_d^{mn}(A,B)$ exceeds ${1\over 2}(n-m)$ when $\Delta A$ and $\Delta B$ are not dual.
Consequently, we can determine that $E_d^{mn}(A,B) > {1\over 2}m+{1\over 2}n+{1\over 2}(n-m) = n$.

Next, let us consider the case in which the row and/or column domain-wall conditions are not fulfilled.
If a row of $A$ is not a domain-wall vector, the numbers of 1/$-1$s in $\Delta A$ are
$(1+k)$/$k$ for some $k\geq 1$.
For example, if a row of $A$ is $11001011$, the corresponding row of $\Delta A$ would be $0010(-1)1(-1)100$, containing three $1$s and two $-1$s.
Based on this observation, we can assign non-negative integers $k_A$ and $k_B$ such that
the numbers of 1/$-1$s in $\Delta A$ and $\Delta B$ are 
$(m+k_A)$/$k_A$ and $(n+k_B)$/$k_B$, respectively.
It follows that $\Delta A$ and $\Delta B$ have $m+2k_A$ and $n+2k_B$ non-zero elements, respectively.
Since we assume that the row and/or column domain-wall conditions are not satisfied, at least one of $k_A$ and $k_B$ must be greater than or equal to 1.
Clearly, the first and second summation terms take values of  ${1\over 2}m+k_A$ and ${1\over 2}n+k_B$, respectively.
The third summation term must be at least
\begin{align}
\left|({1\over 2}m+ {1\over 2}k_A)-({1\over 2}n+ {1\over 2}k_B)\right|+\left|{1\over 2}k_A-{1\over 2}k_B\right|.
\end{align}

We can now analyze the value of $E_d^{mn}(A,B)$ in the following cases:
\begin{enumerate}[label=,labelsep=12mm]
\item [Case~1:]  $({1\over 2}m+ {1\over 2}k_A)\leq ({1\over 2}n+ {1\over 2}k_B)$
\end{enumerate}

The lower bound of $E_d^{mn}(A,B)$ can be evaluated by the sum of the three summation terms as follows:\newpage
\begin{align}
E_d^{mn}(A,B)&\geq({1\over 2}m+k_A)+({1\over 2}n+k_B)+{\left(({1\over 2}n+ {1\over 2}k_B)-({1\over 2}m+ {1\over 2}k_A)\right)}\nonumber\\
  &=n+{1\over 2}k_A+{3\over 2}k_B>n.
\end{align}

\begin{enumerate}[label=,labelsep=12mm]
\item [Case~2:] $({1\over 2}m+ {1\over 2}k_A)> ({1\over 2}n+ {1\over 2}k_B)$
\end{enumerate}

Similarly, $E_d^{mn}(A,B)$ can be evaluated by the sum of the first and second summation terms as follows: 
\begin{align}
E_d^{mn}(A,B)&>({1\over 2}m+k_A)+({1\over 2}n+k_B) >({1\over 2}n+ {1\over 2}k_B)+({1\over 2}n+k_B) >n.
\end{align}

Thus, this completes the proof of the lemma.
\end{proof}

Since we have shown that $E_d^{mn}(A,B)=n$ if the three conditions are satisfied, this lemma implies the following theorem:
\begin{thm}
$E_d^{mn}(A,B)$ takes the minimum value of $n$ if and only if $A$ and $B$ satisfy the row domain-wall, column domain-wall,
and dual-permutation conditions are satisfied.
\end{thm}

Next, we will evaluate the value of $H_d^{mn}(A,B)$ when the three conditions are satisfied.
The first summation term and second summation term take minimum values of $2m$ and $2n$, respectively, if and only if the row and column domain-wall conditions are satisfied.
Additionally, if the dual-permutation condition is satisfied, the third summation term takes a value of $2(n-m)$.
Hence, $H_d^{mn}(A,B)$ is equal to $2m+2n+2(n-m) = 4n$ if the row domain-wall, column domain-wall, and dual-permutation conditions are all satisfied.
Similarly, we can prove that $H_d^{mn}(A,B)$ is larger than $4n$ if at least one of the three conditions is not satisfied.
Thus, we can state the following theorem:
\begin{thm}
$H_d^{mn}(A,B)$ takes the minimum value of $4n$ if and only if $A$ and $B$ satisfy the row domain-wall, column domain-wall, and dual-permutation conditions.
\end{thm}
This theorem establishes the condition for $H_d^{mn}(A,B)$ to attain its minimum value and confirms that it occurs when the row domain-wall, column domain-wall, and dual-permutation conditions are satisfied.

The readers should refer to Table~\ref{tab:partial-stat}, which provides the features of kernels $E_d^{mn}(A,B)$ and $H_d^{mn}(A,B)$.
It can be observed that the number of quadratic terms in these models is given by $6mn-4m-4n$.
They utilize only a quadratic number of quadratic terms, whereas kernels $E_1^{mn}(A,B)$ and $H_1^{mn}(A,B)$ require a cubic number of quadratic terms.
Additionally, the maximum absolute value of the linear term coefficients in $H_d^{mn}(A,B)$ is only $2$, while $H_1^{mn}(A,B)$ requires large coefficients of $m+n-3$.

\begin{table}[H]
\caption{QUBO/Ising kernels for generating a partial permutation selecting $m$ numbers from $n$ numbers.}
\label{tab:partial-stat}
	\begin{adjustwidth}{-\extralength}{0cm}
		\setlength{\cellWidtha}{\fulllength/4-2\tabcolsep+0in}
		\setlength{\cellWidthb}{\fulllength/4-2\tabcolsep-0in}
		\setlength{\cellWidthc}{\fulllength/4-2\tabcolsep+0in}
		\setlength{\cellWidthd}{\fulllength/4-2\tabcolsep-0in}
		\scalebox{1}[1]{\begin{tabularx}{\fulllength}{>{\raggedright\arraybackslash}m{\cellWidtha}>{\centering\arraybackslash}m{\cellWidthb}>{\centering\arraybackslash}m{\cellWidthc}>{\centering\arraybackslash}m{\cellWidthd}>{\centering\arraybackslash}m{\cellWidthd}}
			\toprule
                     &              & \multicolumn{2}{c}{\textbf{Domain-Wall}}\\
\textbf{Encoding Type}  & \textbf{One-Hot} & \textbf{Dual-Matrix} & \textbf{Extended}\\
\midrule
Bit/qubit count &$mn$ & $2mn-m-n$ &$3mn-m-n$\\
\midrule
 QUBO models &&&\\
Quadratic formula &$E_1^{mn}(A,B)$ & $E_d^{mn}(A,B)$ & {$E_e^{mn}(A,B,X)$}\\ 
Linear term count &  $mn$& $2mn-2m-2n$ & $3mn-3m-2n+1$\\
Linear term coefficients & $-1$ & $+2$ & $-1,+1,+2,+3$\\
Quadratic term count &  ${1\over 2}m^2n+{1\over 2}mn^2-mn$ & {$6mn-6m-6n+4$} 
 &{$6mn-4m-4n$}\\
Quadratic term coefficients & $+1,+2$ & $-2,-1,+1$ & $-3,-2,-1,+1,+2$\\
Diameter &  2 & $m+n-3$ & $m+n$\\
Offset & $m$ & $m+n-1$ & ${3\over 2}m+{1\over 2}n$\\
Optimal energy &$0$ & $n$ & ${1\over 2}m+{1\over 2}n$\\
\bottomrule
\end{tabularx}}
\end{adjustwidth}
\end{table}

\begin{table}[H]\ContinuedFloat
\caption{{\em Cont.}}
\begin{adjustwidth}{-\extralength}{0cm}
		\setlength{\cellWidtha}{\fulllength/4-2\tabcolsep+0in}
		\setlength{\cellWidthb}{\fulllength/4-2\tabcolsep-0in}
		\setlength{\cellWidthc}{\fulllength/4-2\tabcolsep+0in}
		\setlength{\cellWidthd}{\fulllength/4-2\tabcolsep-0in}
		\scalebox{1}[1]{\begin{tabularx}{\fulllength}{>{\raggedright\arraybackslash}m{\cellWidtha}>{\centering\arraybackslash}m{\cellWidthb}>{\centering\arraybackslash}m{\cellWidthc}>{\centering\arraybackslash}m{\cellWidthd}>{\centering\arraybackslash}m{\cellWidthd}}
			\toprule
                     &              & \multicolumn{2}{c}{\textbf{Domain-Wall}}\\
\textbf{Encoding Type}  & \textbf{One-Hot} & \textbf{Dual-Matrix} & \textbf{Extended}\\
\midrule
Ising models &&&\\
Quadratic formula &  $H_1^{mn}(A,B)$  & $H_d^{mn}(A,B)$ & {$H_e^{mn}(A,B,S)$} \\ 
Linear term count & $mn$ & $2m+2n$ & $mn+2m+2n$\\
Linear term coefficients & {$m+n-3$}  & {$-2,+2$} & {$-3,-1,+1,+2,+3,+4$}\\
Quadratic term count & ${1\over 2}m^2n+{1\over 2}mn^2-mn$ &  {$6mn-6m-6n+4$} & {$6mn-4m-4n$}\\
Quadratic term coefficients & $+1$  & $-2,-1,+1$ & $-3,-2,-1,+1,+2$\\
Diameter & $2$ & $m+n-3$ & $m+n$\\
Offset & ${1\over 2}m^2n+{1\over 2}mn^2-2mn+2m$ & $4mn-4$ & $8mn-4m-2n$\\
Optimal Hamiltonian & $0$ & $4n$ & $2m+2n$\\
\bottomrule
\end{tabularx}}
	\end{adjustwidth}
\end{table}

\section{Dual-Matrix Domain-Wall Technique Extended with a One-Hot Matrix}
\label{sec:extended}
When the QUBO kernels designed by the dual-matrix domain-wall technique are involved in QUBO models for solving permutation-based combinatorial optimization problems, $\Delta A=(\Delta a_{i,j})$ is used to compute the objective function of the optimization problems.
For example, when checking if both $\pi(i)=j$ and $\pi(i')=j'$ are satisfied, a quadratic term $\Delta a_{i,j}\Delta a_{i',j'}$ is employed.
However, since $\Delta a_{i,j}$s are not variables, they are expanded using the variable $a_{i,j}$s as shown below:
\begin{align}
\Delta a_{i,j}\Delta a_{i',j'}&=(a_{i,j-1}-a_{i,j})(a_{i',j'-1}-a_{i',j'})\nonumber\\
 &=a_{i,j-1}a_{i',j'-1}-a_{i,j-1}a_{i',j'}-a_{i,j}a_{i',j'-1}+a_{i,j}a_{i',j'},
\label{eq:four}
\end{align}

Since a single quadratic term $\Delta a_{i,j}\Delta{a_{i',j'}}$ is expanded to the four quadratic terms involving the variable $a_{i,j}$s,
the resulting QUBO model obtained through reduction contains numerous quadratic terms.
To mitigate this issue, we introduce a matrix $X=(x_{i,j})$ with $n\times n$-bit variables, which satisfies $X=\Delta A$, into our QUBO kernel. 
Similarly, a one-hot matrix $S=(s_{i,j})$ with $n\times n$-qubit variables
is added to the Ising kernel.
The matrix ensures that $s_{i,j}=\Delta a_{i,j}-1$ for all $i$ and $j$, because each $\Delta a_{i,j}$ takes 0 or 2.
Since these matrices $X$ and $S$ store a permutation as a one-hot encoding, we refer to them as \emph{one-hot matrices}

To extend our dual-matrix domain-wall technique for one-hot matrices, we introduce an additional condition along with the row domain-wall condition, column domain-wall condition, and dual-permutation condition.
This new condition is defined as follows:

\textbf{One-hot permutation condition}. Each row of matrix $X$/$S$ represents the permutation corresponding to matrix $A$
in the form of a one-hot vector.

We begin by designing a QUBO model that satisfies these four conditions.
The mathematical expression for this model is as follows:
\begin{align}
E_e^{nn}(A,B,X) &={1\over 2}\sum_{i=0}^{n-1}\sum_{j=0}^{n-1} (\Delta a_{i,j})^2+{1\over 2}\sum_{i=0}^{n-1}\sum_{j=0}^{n-1}(\Delta b_{i,j})^2\nonumber\\
 &\quad +{1\over 2}\sum_{i=0}^{n-1}\sum_{j=0}^{n-1} (x_{i,j}-\Delta a_{i,j})^2+{1\over 2}\sum_{i=0}^{n-1}\sum_{j=0}^{n-1} (x_{i,j}-\Delta b_{i,j})^2
\label{eq:Ee_nn}
\end{align}

Similarly to the previous model, the first and second summation terms achieve the minimum value of $n\over 2$ when the row and column domain-wall conditions are satisfied, respectively.
The third and fourth summation terms reach the minimum value of 0 if both the dual-permutation condition and the one-hot permutation condition are met. Therefore, matrix $X$ stores a permutation if and only if $E_e^{nn}(A,B,X)$ reaches its minimum value of $n$.

The same technique can also be applied to the Ising model. 
In this case, we introduce matrix $S=(s_{i,j})$ of size $n\times n$, where $s_{i,j}+1=\Delta a_{i,j}$ holds for all $i$ and $j$.
The Ising model $H_e^{nn}(A,B,S)$ can be designed as follows:
\begin{align}
H_e^{nn}(A,B,S) &={1\over 2}\sum_{i=0}^{n-1}\sum_{j=0}^{n-1} (\Delta a_{i,j})^2+{1\over 2}\sum_{i=0}^{n-1}\sum_{j=0}^{n-1}(\Delta b_{i,j})^2\nonumber\\
 &\quad  +{1\over 2}\sum_{i=0}^{n-1}\sum_{j=0}^{n-1} ((s_{i,j}+1)-\Delta a_{i,j})^2+{1\over 2}\sum_{i=0}^{n-1}\sum_{j=0}^{n-1} ((s_{i,j}+1)-\Delta b_{i,j})^2
\label{eq:He_nn}
\end{align}

Similarly to the QUBO model,  the four conditions are satisfied, and matrix $S$ stores a one-hot encoding of a permutation if and only if $H_e^{nn}(A,B,S)$ reaches its minimum value of $4n$.
Optimal solutions of $H_e^{nn}(A,B,S)$ correspond to a one-hot encoding stored in \mbox{matrix $S$.}
{The features of QUBO/Ising models obtained by expanding $E_e^{nn}(A,B,X)$ and $H_e^{nn}(A,B,S)$} {are presented in Table~\ref{tab:permutation-stat}.}

We further apply the same technique for a partial permutation.
However, a straightforward modification of $E_e^{nn}(A,B,X)$/$H_e^{nn}(A,B,S)$ to fit them to a partial permutation may not
generate QUBO/Ising models correctly.
The optimal solutions of such QUBO/Ising models may not satisfy the four conditions.
We modify their fourth summation term and obtain the following QUBO model $E_e^{mn}(A,B,X)$:
\begin{align}
E_e^{mn}(A,B,X) &={1\over 2}\sum_{i=0}^{m-1}\sum_{j=0}^{n-1} (\Delta a_{i,j})^2+{1\over 2}\sum_{i=0}^{m-1}\sum_{j=0}^{n-1}(\Delta b_{i,j})^2\nonumber\\
 &\quad +\sum_{i=0}^{m-1}\sum_{j=0}^{n-1} (x_{i,j}-\Delta a_{i,j})^2+\sum_{i=0}^{m-1}\sum_{j=0}^{n-1} x_{i,j}(1-\Delta b_{i,j})
\label{eq:Ee_mn}
\end{align}

Note that the third and fourth summation terms do not have a coefficient of $1\over 2$, because the terms in the expanded formula
will have a non-integer coefficient if they have a coefficient of $1\over 2$.
The first and second summation terms take minimum values of ${1\over 2}m$ and ${1\over 2}n$, respectively,
if and only if the row and column domain-wall conditions are satisfied.
The third summation term takes the minimum value of 0 if $X$ and $\Delta A$ are equal.
The fourth summation term takes the minimum value of 0 if $x_{i,j}=0$ or $\Delta b_{i,j}=1$ holds for all $i$ and $j$.
Thus, $E_e^{mn}(A,B,X)={1\over 2}m+{1\over 2}n$ if the four conditions are satisfied.
We present the following lemma to prove that the ``if'' in the previous statement is ``if and only if'':
\begin{lem}
\label{lem:Ed_mn}
$E_e^{mn}(A,B,X)$ is larger than ${1\over 2}(m+n)$ if matrices $A$, $B$, and $X$ do not satisfy at least one of the row domain-wall, column domain-wall, dual-permutation, and one-hot permutation conditions.
\end{lem}
\begin{proof}
For any $A$, $B$, and $X$, the first, second, third, and fourth summation terms in $E_e^{mn}(A,B,X)$ are at least ${1\over 2}m$, ${1\over 2}n$, 0, and 0, respectively.
If the row/column domain-wall conditions are not satisfied, the first and second summation terms are larger than ${1\over 2}m$ and ${1\over 2}n$, respectively.
Hence, if the row and/or the column domain-wall conditions are not satisfied, $E_e^{mn}(A,B,X)>{1\over 2}(m+n)$.

Suppose that the row and column domain-wall conditions are satisfied.
Clearly, matrices $A$ and $B$ have $m$ and $n$~1s, respectively, and all the other elements are 0.
Thus, the first and second summation terms take minimum values of ${1\over 2}m$ and ${1\over 2}n$, respectively.
If the dual-permutation condition is not satisfied, that is, $\Delta A$ and $\Delta B$ are not dual, then
there exist $i$ and $j$ such that $\Delta a_{i,j}=1$ and $\Delta b_{i,j}=0$.
For such an $i$ and $j$, either {$(x_{i,j}-\Delta a_{i,j})^2$} or {$x_{i,j}\cdot(1-\Delta b_{i,j})$} must be 1.  
Thus, at least one of the third or fourth summation terms is larger than 0.
If the one-hot permutation condition is not satisfied, 
then $X$ and $\Delta A$ are not equal, and the third summation term is larger than 0.
Thus, $E_e^{mn}(A,B,X)>{1\over 2}m+{1\over 2}n$ holds.
This completes the proof of Lemma~\ref{lem:Ed_mn}.
\end{proof}
Since we have proved that $E_e^{mn}(A,B)={1\over 2}m+{1\over 2}n$ if the four conditions are satisfied, we can state the following theorem:
\begin{thm}
$E_e^{mn}(A,B,X)$ takes the minimum value of ${1\over 2}m+{1\over 2}n$ if and only if $A$ and $B$ satisfy
the row domain-wall, column domain-wall, dual-permutation, and one-hot permutation conditions.
\end{thm}

We can extend the Ising model for a partial permutation to have a one-hot matrix $S$ in the same way as follows:
\begin{align}
H_e^{mn}(A,B,S) &={1\over 2}\sum_{i=0}^{m-1}\sum_{j=0}^{n-1} (\Delta a_{i,j})^2+{1\over 2}\sum_{i=0}^{m-1}\sum_{j=0}^{n-1}(\Delta b_{i,j})^2\nonumber\\
 &\quad  +\sum_{i=0}^{m-1}\sum_{j=0}^{n-1} ((s_{i,j}+1)-\Delta a_{i,j})^2+\sum_{i=0}^{m-1}\sum_{j=0}^{n-1} (s_{i,j}+1)\cdot(2-\Delta b_{i,j})
\label{eq:He_mn}
\end{align}

If the four conditions are satisfied, then the first, second, third, and fourth summation terms
take the values $2m$, $2n$, 0, and 0, respectively.
Thus, $H_e^{mn}(A,B,S)=2m+2n$ holds when the four conditions are satisfied.
Similarly, we can prove that $H_e^{mn}(A,B)$ is larger than $2m+2n$ if at least one of the four conditions is not satisfied.
Thus, we can state the following theorem:
\begin{thm}
$H_e^{mn}(A,B,S)$ takes the minimum value of $2m+2n$ if and only if $A$, $B$, and $S$ satisfy the row domain-wall, column domain-wall, dual-permutation, and
one-hot permutation conditions.
\end{thm}

The reader should refer to Table~\ref{tab:partial-stat}, which provides the features of $E_e^{mn}(A,B,X)$ and $H_e^{mn}(A,B,S)$.
It can be observed that the number of quadratic terms in these kernels is given by $6mn-4m-4n$, while $E_d^{mn}(A,B)$ and $H_d^{mn}(A,B)$ have $6mn-6m-6n+4$ quadratic terms.
Therefore, the introduction of one-hot matrices results in a relatively small increment in the quadratic term counts.
This reduction in the number of quadratic terms highlights the effectiveness of utilizing one-hot matrices $X$ and $S$ in the QUBO and Ising models for partial permutations.
By employing these techniques, we can significantly enhance the efficiency of solving optimization problems associated with partial permutations

\section{Applications of QUBO/Ising Kernels for Generating Permutations in Combinatorial Optimization Problems}
\label{sec:ppp}
In this section, we begin by introducing \emph{{the Particle Placement Problem (PPP)}
}, which serves as a fundamental problem that many permutation-based combinatorial optimization problems can be reduced to with ease.
We demonstrate that the PPP can be effectively reduced to QUBO/Ising models with kernels for generating permutations.
Furthermore, we show that several permutation-based combinatorial optimization problems, such as the quadratic assignment problem (QAP), traveling salesman problem (TSP), and the sub-graph isomorphism problem, can be equivalently reduced to the PPP.
Consequently, these combinatorial optimization problems can be reduced to QUBO/Ising models through the PPP.

\subsection{The Particle Placement Problem (PPP)}
Suppose we have $m$ particles with IDs $0, 1, \ldots, m-1$ and $n$ positions with IDs $0, 1,\ldots,$\linebreak $ n-1$, where $m \leq n$ is satisfied.
Let us consider the PPP, which involves placing $m$ particles in $n$ positions without conflicts.
\emph{The particle placement} is determined by injection $\pi: \{0, 1, \ldots, m-1\} \rightarrow \{0, 1, \ldots, n-1\}$,
where each particle $i$ ($0 \leq i \leq m-1$) is placed in position $\pi(i)$.
The problem aims to find an optimal placement $\pi$ that minimizes the total energy, which is defined by two types of energies as follows:

{\bf Potential \boldmath{$P_{i,j}$}}: the energy that occurs when particle $i$ is placed in position $j$ ($0 \leq i \leq m-1$ and $0 \leq j \leq n-1$).

{\bf Interaction \boldmath{$I_{i,j,i',j'}$}}: the energy that occurs between two particles $i$ and $i'$ when they are placed in positions $j$ and $j'$, respectively, where $(i,j,i',j') \in Q(m,n)$.

Here, $Q(m,n)$ is a set of all consistent quartets $(i,j,i',j')$ satisfying $0 \leq i < i' \leq m-1$, $0 \leq j,j' \leq m-1$, and $j \neq j'$.
Both potential $P_{i,j}$ and interaction $I_{i,j,i',j'}$ can be any integer, including negative values.
The total energy of $\pi$ is the sum of all potentials and interactions, given by the equation
\begin{align}
{\it PPP}(\pi) = \sum_{i=0}^{m-1} P_{i,\pi(i)} + \sum_{i=0}^{m-2} \sum_{i'=i+1}^{m-1} I_{i,\pi(i),i',\pi(i')}
\end{align}

\emph{The particle placement problem (PPP)} aims to find an optimal placement $\pi$ with the minimum ${\it PPP}(\pi)$ among all possible particle placements.

We will demonstrate that the total energy of $\pi$ can be computed using the corresponding binary one-hot matrix, and we can reduce the PPP to a QUBO problem.
In this sub-section, we focus on the PPP with $m$ particles and $n$ locations, where $m<n$ and $\pi$ represents a partial permutation.
However, it is worth noting that the same technique can be applied to the PPP with $n$ particles and $n$ locations, as it is relatively straightforward.

Suppose that an $m\times n$-bit matrix $X$ stores a partial permutation $\pi$ as one-hot vectors.
In other words, $x_{i,j}=1$ holds if and only if $\pi(i)=j$.
Let ${\it PPP}(X)$ be a quadratic formula defined as follows:
\begin{align}
{\it PPP}(X) &= \sum_{i=0}^{m-1}\sum_{j=0}^{n-1}P_{i,j}x_{i,j} +\sum_{(i,j,i',j')}^{Q(m,n)}   I_{i,j,i',j'} x_{i,j}x_{i',j'}.
\label{eq:Energy}
\end{align}

Clearly, $x_{i,j}=1$ if and only if $\pi(i)=j$, and $x_{i,j}x_{i',j'}=1$ if and only if $\pi(i)=j$ and $\pi(i')=j'$.
Thus, we have ${\it PPP}(X)={\it PPP}(\pi)$.
Since ${\it PPP}(X)$ is a quadratic formula of $X$, we can design a QUBO model using a QUBO kernel $E_1^{mn}(X)$ corresponding to the PPP as follows:
\begin{align}
E_1^{\it PPP}(X)&= \lambda E_1^{mn}(X) +{\it PPP}(X),
\end{align}
where $\lambda$ is a parameter large enough to prioritize $E_1^{mn}(X)$ and guarantee that
$X$ is a one-hot vector.
By finding an $n\times n$-bit matrix $X$ that minimizes $E_1^{\it PPP}(X)$,
we can obtain the optimal solution $\pi$ of the PPP corresponding to such an $X$.

We can use the other QUBO kernels $E_d^{mn}(A,B)$ and $E_e^{mn}(A,B,X)$ instead of $E_1^{mn}(X)$: 
\begin{align}
E_d^{\it PPP}(A,B)&= \lambda E_d^{mn}(A,B) +{\it PPP}(\Delta A)\\
E_e^{\it PPP}(A,B,X)&= \lambda E_e^{mn}(A,B,X) +{\it PPP}(X)
\end{align}

Similarly, by finding their optimal solutions, we can obtain the optimal solution $\pi$ of the PPP.

For an $m\times n$-qubit matrix $S=(s_{i,j})$, let $S'=(s'_{i,j})$ denote the matrix of the same size satisfying
$s'_{i,j}={s_{i,j}+1}$ for all $i$ and $j$.
Clearly, $s'_{i,j}=0$/$2$ if $s_{i,j}=-1$/$+1$, respectively.
Thus, for ${\it PPP}(S')$ computed by the mathematical expression in Equation~(\ref{eq:Energy}),
${\it PPP}(S')=4{\it PPP}(\pi)$ is satisfied if $S$ stores permutation $\pi$ as one-hot vectors.
{Using this fact, we can design Ising models for solving the PPP as follows:}
\begin{align}
H_1^{\it PPP}(S)&= {\lambda H_1^{mn}(S)} +{\it PPP}(S')\\ 
H_d^{\it PPP}(A,B)&= \lambda H_d^{mn}(A,B) +{\it PPP}(\Delta A)\\
H_e^{\it PPP}(A,B,S)&= {\lambda H_e^{mn}(A,B,S)} +{\it PPP}(S') 
\end{align}

It should be clear that the optimal solutions of these Ising models give the optimal solutions of the PPP if
$\lambda$ is large enough to prioritize the Ising kernels for generating partial permutations.

The PPP can have a total of $mn$ potentials and ${1\over 2}m^2n^2-{1\over 2}m^2n-{1\over 2}mn^2+{1\over 2}mn$ interactions.
When the PPP is converted to the equivalent QUBO/Ising models, non-zero interactions may lead to quadratic terms.
 In essence, the total number of quadratic terms in the QUBO/Ising models for solving the PPP is the sum of the number of terms produced by a kernel generating permutations and the number of quadratic terms for non-zero interactions in the PPP.
As demonstrated in this paper, our dual-matrix domain-wall technique significantly reduces the count of quadratic terms.
However, if the number of quadratic terms derived from non-zero interactions in the PPP is much larger than that for the kernel generating permutations, the advantage of our dual-matrix domain-wall technique becomes limited.
It is worth noting that both the conventional one-hot technique and the all-different domain-wall technique require a cubic number of quadratic terms,
whereas our dual-matrix domain-wall technique only uses a quadratic number of quadratic terms.
If the PPP has full non-zero interactions, it would result in a number of quadratic terms that is proportional to the fourth power.
In such a scenario, the resulting QUBO/Ising models would have a fourth power number of quadratic terms, and the reduction effect achieved by the dual-matrix domain-wall technique for the kernel terms would be relatively small.
To evaluate the effect of the dual-matrix domain-wall technique on specific combinatorial optimization problems,
we will present several examples and assess the number of non-zero interactions in the PPP.

\subsection{Combinatorial Optimization Problems that Can Be Reduced to the PPP}
\textls[-15]{This subsection demonstrates that several combinatorial optimization problems, namely} the quadratic assignment problem (QAP), the traveling salesman problem (TSP), and the sub-graph isomorphism problem, which are known to be NP-hard~\cite{GJ79}, can be effectively reduced to the PPP. 
Consequently, the instances of these problems can be transformed into equivalent QUBO/Ising models.
Additionally, we illustrate that the maximum weight matching problem can also be reduced to the PPP. While this problem is not NP-hard and can be solved in polynomial time~\cite{Edmonds65}, there is currently no known efficient parallel algorithm for general graphs~\cite{Eppstein21}.
Consequently, if a highly powerful QUBO/Ising solver becomes available in the future,
it could be advantageous to solve this problem by utilizing the reduction to QUBO/Ising models.

\subsubsection{The Quadratic Assignment Problem (QAP)}
The quadratic assignment problem (QAP)~\cite{Koopmans57} aims
to find the optimal arrangement $\pi$ of $n$ facilities to $n$ positions.
This problem involves two types of instance parameters: flows and distances.
Let $f_{i,i'}$ and $d_{j,j'}$ denote the flow from facility $i$ to $i'$ ($0\leq i,i'\leq n-1$) and the distance from $j$ to $j'$ ($0\leq j,j'\leq n-1$), respectively.
The objective of QAP is to find a permutation $\pi$ that minimizes the total logistics cost, defined by the following formula:
\begin{align}
{{\it QAP}(\pi)} &=\sum_{i=0}^{n-1}\sum_{i'=0}^{n-1}f_{i,i'}\cdot d_{\pi(i),\pi(i')}.  
\end{align}

To reduce this problem to the PPP,
we set $I_{i,j,i',j'}=f_{i,i'}\cdot d_{j,j'}+f_{i',i}\cdot d_{j',j}$ for all $i$, $i'$, $j$, and $j'$ ($(i,j,i',j')\in Q(n,n)$)
and set $P_{i,j}=f_{i,i}\cdot d_{j,j}$ for all $i$ and $j$ ($0\leq i,j\leq n-1$).
Since {${\it QAP}(\pi)={\it PPP}(\pi)$} holds for all permutation $\pi$, any QAP instance can be reduced to the PPP.
The reduced PPP has up to ${1\over 2}n^4-n^3+{1\over 2}n^2$ non-zero interactions and $n^2$ non-zero potentials.
If not all flows are non-zero, the number of non-zero interactions may be reduced.
Let $k$ be the number of pairs of facilities $i$ and $i'$ ($0\leq i<i'\leq n-1$) such that at least one of $f_{i,i'}$ or $f_{i',i}$ is non-zero.
Clearly, $k$ is at most ${1\over 2}n^2-{1\over 2}n$.
In this case, the reduced PPP has only ${1\over 2}n^2k-{1\over 2}nk$ non-zero interactions.

\subsubsection{The Traveling Salesman Problem (TSP)}
The traveling salesman problem (TSP), which aims to
find the shortest route to visit all $n$ cities with IDs 0, 1, $\ldots$, $n-1$ can be reduced to the PPP as follows.
Let $d_{j,j'}$ ($0\leq j,j' \leq n-1$) denote the distance between two cities $j$ and $j'$.
We use a permutation $\pi$ for representing a route,
where the $i$-th visited city ($0\leq i\leq n-1$) corresponds to the city with ID $\pi(i)$.
The objective of the TSP is to find a permutation $\pi$ that minimizes the total route length computed using the following formula:
\begin{align}
{{\it TSP}(\pi)} &=\sum_{i=0}^{n-1} d_{\pi(i),\pi((i'+1)\bmod n)}. 
\label{eq:tsp}
\end{align}

To reduce the TSP to the PPP, we define $I_{i,j,i',j'}=d_{j,j'}$ for all $i$, $i'$, $j$, and $j'$ if $i'=(i+1) \bmod n$.
For all other $I_{i,j,i',j'}$ and $P_{i,j}$, we set them to 0.
This reduction allows us to observe that {${\it TSP}(\pi)={\it PPP}(\pi)$} holds for all permutations $\pi$.  
Therefore, it is possible to reduce any TSP instance to the PPP.
The reduced PPP has $n^3-n^2$ non-zero interactions.

We can consider the traveling salesman problem (TSP) for a weighted undirected graph in which each edge is assigned a weight value. 
The goal of the TSP is to find a Hamiltonian cycle in the graph with the minimum total weight that visits all nodes.
Figure~\ref{fig:tsp40} shows an example of a 40-node weighted graph.
We assume that the weight of each edge is its length in the figure. The figure also depicts the optimal solution of the TSP.
Note that if an input graph has no Hamiltonian cycle, the TSP has no feasible solution.

\begin{figure}[H]
\includegraphics[scale=0.6]{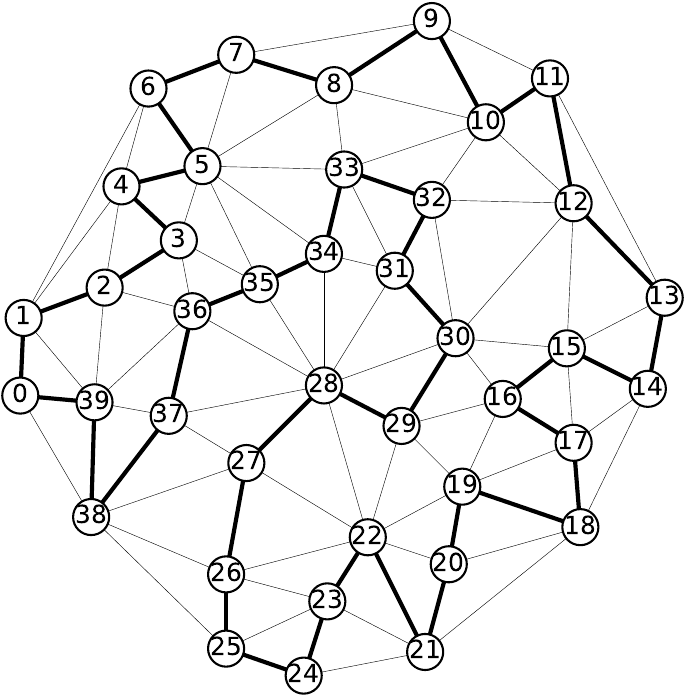}
\caption{A 40-node weighted graph with the optimal TSP solution.}
\label{fig:tsp40}
\end{figure}

Suppose that we have a weighted graph with nodes 0, 1, $\ldots$, $n-1$, and let $d_{j,j'}$ ($0\leq j,j' \leq n-1$) be the weight of an edge $(j,j')$.
If the graph does not have an edge $(j,j')$, we assume $d_{j,j'}=+\infty$ to prevent selecting $(j,j')$ as an edge of a cycle.
The total route of a node permutation $\pi$ is computed using Equation~(\ref{eq:tsp}), based on the defined $d_{j,j'}$ values.
For the real implementation, we set $d_{j,j'}={\it BIG}$ instead of $+\infty$, where ${\it BIG}$ is sufficiently larger than the maximum edge weight.
An edge $(j,j')$ with $d_{j,j'}={\it BIG}$ is selected only if the TSP has no feasible solution.

We will introduce the technique for reducing the TSP into the PPP, so that the reduced PPP becomes sparse. For this purpose, we apply a fixed bias of $-{\it BIG}$,
so that $I_{i,j,i',j'}=d_{j,j'}-{\it BIG}$ for all $i$, $i'$, $j$, and $j'$ if $i'=(i+1) \bmod n$. For all other $I_{i,j,i',j'}$ and $P_{i,j}$, we set them to 0.
This reduction allows us to observe that {${\it TSP}(\pi)={\it PPP}(\pi)+n\cdot{\it BIG}$}  holds for all permutations $\pi$.  
Additionally, $I_{i,j,i',j'}$ is non-zero only if the graph has an edge $(j,j')$ and $i'=(i+1) \bmod n$ holds.
Thus, the PPP becomes so sparse that it has only $2en$ non-zero interactions, where $e$ is the number of edges in the graph.
For instance, if no two edges intersect each other, as illustrated in Figure~\ref{fig:tsp40}, the graph has at most $3n-6$ edges.
Hence, the reduced PPP has at most $2(3n-6)n=6n^2-12$ interactions, which is much smaller than the original TSP with $n^3-n^2$ interactions.

\subsubsection{The Sub-graph Isomorphism Problem}
\textls[-35]{Let $G=(V,E)$ and $G'=(V',E')$ be guest and host graphs
such that $V=\{0, 1, \ldots, m-1\}$} and $V'=\{0, 1, \ldots, n-1\}$.
The sub-graph isomorphism problem aims to find a permutation $\pi$ such that
$(\pi(i),\pi(i'))\in E'$ holds for all $(i,i')\in E$.
If such a permutation $\pi$ exists,  it implies that $G$ is a sub-graph of $G'$.
To reduce the sub-graph isomorphism problem to the PPP,
we assign $-1$ to $I_{i,j,i',j'}$ and $I_{i,j',i',j}$ if both $(i,i')\in E$ and $(j,j')\in E'$ hold,
and assign $0$ to all other $I_{i,j,i',j'}$ as well as all $P_{i,j}$.
It is evident that $I_{i,\pi(i),i',\pi(i')}=-1$ holds for all edges $(i,i')$ in $G$ if $\pi$ is a permutation that proves $G$ is a sub-graph of $G'$.
Therefore, {${\it PPP}(\pi)=-|E|$} if such a permutation exists. 
Consequently, any instance of the sub-graph isomorphism problem can be reduced to the PPP.
The reduced PPP has $2\cdot |E|\cdot |E'|$ non-zero interactions.
If $G$ and $G'$ are connected graphs, then $m-1\leq |E|\leq {1\over 2}m^2-{1\over 2}m$ and $n-1\leq |E'|\leq {1\over 2}n^2-{1\over 2}n$ hold.
Hence, the reduced PPP has from {$2mn-2m-2n+2$} to ${1\over 2}m^2n^2-{1\over 2}m^2n-{1\over 2}mn^2+{1\over 2}mn$ non-zero interactions. 

\subsubsection{The Maximum Weight Matching Problem}
Let $G=(V,E,w)$ be a weighted graph, where $V=\{0, 1, \ldots, n-1\}$ and $w_{i,j}$ represents the weight of an edge $(i,j)$ ($i<j$) in $E$.
We assume that $w_{i,j}=0$ for all other pairs $(i,j)$ such that $i\geq j$ or $(i,j)$ is not in $E$.
A subset of $E$ is called a matching if none of its edges share a common node.
The maximum weight matching problem aims to find a matching $M$ ($\subseteq E$) with the largest total weight.
We can represent a matching $M$ using a permutation $\pi$ as follows. An edge $(i,j)$ ($\in E$) is in the matching $M$
if both $\pi(i)=j$ and $\pi(j)=i$ hold, and node $i$ is not connected to any edge in $M$ if $\pi(i)=i$.
To transform the maximum weight matching problem into the PPP,
we assign $-w_{i,j}$ to $I_{i,j,j,i}$ for all $(i,j)\in E$, while assigning $0$ to all other $I_{i,j,i',j'}$ as well as all $P_{i,j}$.
Let $\pi$ be a permutation that represents a matching $M$.
Since $w_{i,\pi(i)}=0$ whenever $i\geq \pi(i)$, we can compute the total weight of $M$ using the obtained PPP as follows:
\begin{align}
\sum_{(i,j)}^{M} w_{i,j} &= \sum_{i=0}^{n-1} w_{i,\pi(i)} =-\sum_{i=0}^{n-1} I_{i,\pi(i),\pi(i),i} = -{\it PPP}(\pi).
\end{align}

Therefore, any instance of the maximum weight matching problem can be reduced to the PPP with $|E|$ interactions.
If $G$ is a connected graph, $n-1\leq |E|\leq {1\over 2}n^2-{1\over 2}n$ holds,
and so the number of interactions is from $n-1$ to ${1\over 2}n^2-{1\over 2}n$.

\subsubsection{The Bipartite Maximum Weight Matching Problem}
If an instance of the maximum weight matching problem is restricted to a bipartite graph,
a more efficient reduction to the PPP can be used.
Consider a weighted bipartite graph, denoted as $G=(V,V',E,w)$, where $V=\{0, 1, \ldots, m-1\}$, $V'=\{0, 1, \ldots, n-1\}$,
$E\subseteq V\times V'$, and $w_{i,j}$ represents the weight of an edge $(i,j) \in E$.
We can use a partial permutation $\pi:V\rightarrow V'$ to represent a matching $M$ of $G$ such that $(i,\pi(i)) \in M$.
To transform the maximum weight matching problem for the weighted bipartite graph into the PPP,
we assign $-w_{i,j}$ to $P_{i,j}$ for all $(i,j)\in E$, while assigning $0$ to all other $P_{i,j}$ as well as all $I_{i,j,i',j'}$.
Thus, we can compute the total weight of $M$ using the obtained PPP as follows:
\begin{align}
\sum_{(i,j)}^{M} w_{i,j} &= \sum_{i=0}^{m-1} w_{i,\pi(i)} =-\sum_{i=0}^{m-1}P_{i,\pi(i)} = -{\it PPP}(\pi).
\end{align}

Therefore, any instance of the bipartite maximum weight matching problem can be reduced to the PPP.
The resulting PPP has $|E|$ non-zero potentials but has no non-zero interactions.

\subsubsection{The Quadratic Term Counts of the Ising Model}
This section presents the quadratic term counts for various permutation-based combinatorial optimization problems. We examined the total number of quadratic terms in Ising models and permutation kernels to observe the impact of smaller permutation kernels on the size of Ising models. To facilitate this analysis, we introduced the concept of PPP density, which represents the ratio of non-zero interactions to the maximum number of interactions $|Q(m,n)|={1\over 2}m^2n^2-{1\over 2}m^2n-{1\over 2}mn^2+{1\over 2}mn$.

{
We used PyQUBO~\cite{Pyqubo}, a Python tool designed to generate QUBO/Ising models from mathematical expressions, and SymPy~\cite{Sympy}, a Python library for symbolic mathematics, to validate the accuracy of the QUBO/Ising models. However, this approach, which operates on provided mathematical expressions, has a significant drawback in terms of its substantial time and space requirements. This limitation could make it impractical to solve particularly large problems.
For optimization problems based on specific permutations, an alternative and more efficient approach is to create a custom computer program for generating QUBO/Ising models. This approach can leverage a low-level programming language like C++.
Given the fixed and straightforward nature of the underlying mathematical expressions for such problems, it becomes feasible to design a program that efficiently generates QUBO/Ising models without the need to expand the mathematical expressions.
This process can be achieved with linear time complexity and requires a memory proportional to the size of the models.
}

Figure~\ref{fig:quadratic_term_count} illustrates the quadratic term counts for the following problem instances:
\begin{enumerate}
\item A random instance of the QAP with $n=50$: The reduced PPP consists of all 3,001,250 interactions, yielding a density of $1.00$.
\item A random instance of the QAP with $n=50$ and $10$\% non-zero flows: The reduced PPP contains 301,350 interactions out of 3,001,250 maximum interactions, yielding a density of $0.10$.
\item A random instance of the TSP with $n=100$: The reduced PPP includes 990,000 interactions out of 49,005,000 maximum interactions, yielding a density of 0.020.
\item The TSP for a randomly generated 40-node weighted graph as shown in Figure~\ref{fig:tsp40}: The graph has 104 edges, and the reduced PPP includes 8320 interactions out of 1,216,800 maximum interactions, yielding a density of $0.0068$.
\item The TSP for a randomly generated 300-node weighted graph:  The graph is generated to avoid edge intersections, similar to the 40-node weighted graph. It has 277 edges, and the reduced PPP includes 524,400 interactions out of a maximum of 4,023,045,000 interactions, yielding a density of $0.00013$.
\item A random instance of the sub-graph isomorphism problem for a three-regular guest graph with $m=200$ nodes and a six-regular host graph with $n=400$ nodes: The reduced PPP consists of 720,000 interactions out of 3,176,040,000 maximum interactions, yielding a density of $0.00023$.
\item A random instance of the maximum weight matching problem for a complete graph with $n=300$ nodes: The reduced PPP contains 44,850 interactions out of 4,023,045,000 maximum interactions, yielding a density of $0.000011$.
\item A random instance of the bipartite maximum weight matching problem for a complete bipartite graph with $m=n=300$ nodes: The reduced PPP has no interactions, yielding a density of zero.
\end{enumerate}

The figure illustrates the total number of quadratic terms in the resulting Ising model, referred to as the “model total”, for each problem. Additionally, it displays the number of quadratic terms contributed by the permutation kernels, referred to as the “permutation kernel”, which is also included in the model total.
The quadratic term counts are evaluated for four permutation kernels: the conventional one-hot technique (one-hot), the all-different domain-wall technique (all-different), the dual-matrix domain-wall technique (dual-matrix), and the extended technique with a one-hot matrix (extended).

\begin{figure}[H]
	\begin{adjustwidth}{-\extralength}{0cm}
\begin{tabular}{cc}
\includegraphics[width=8cm]{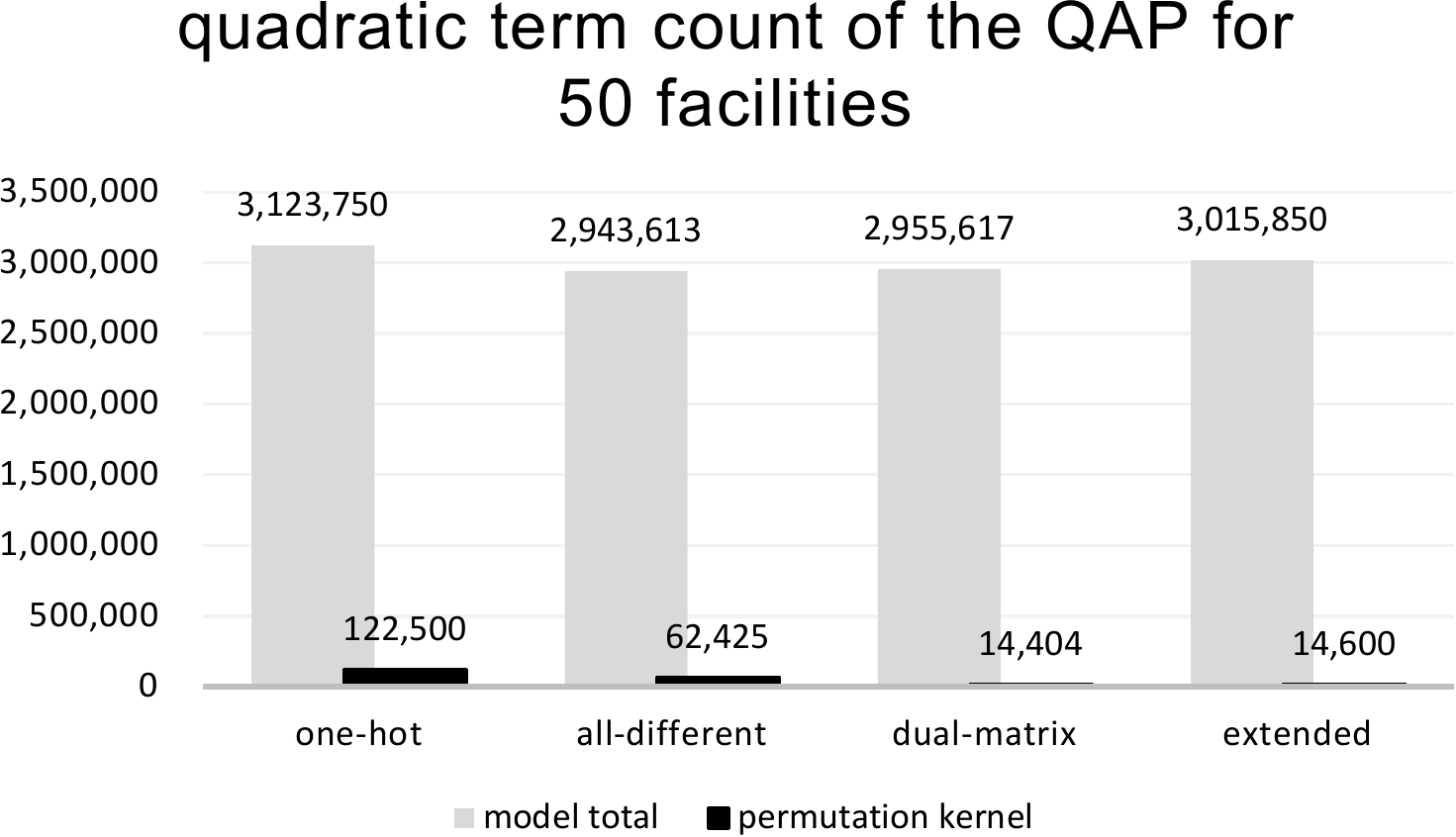}&\includegraphics[width=8cm]{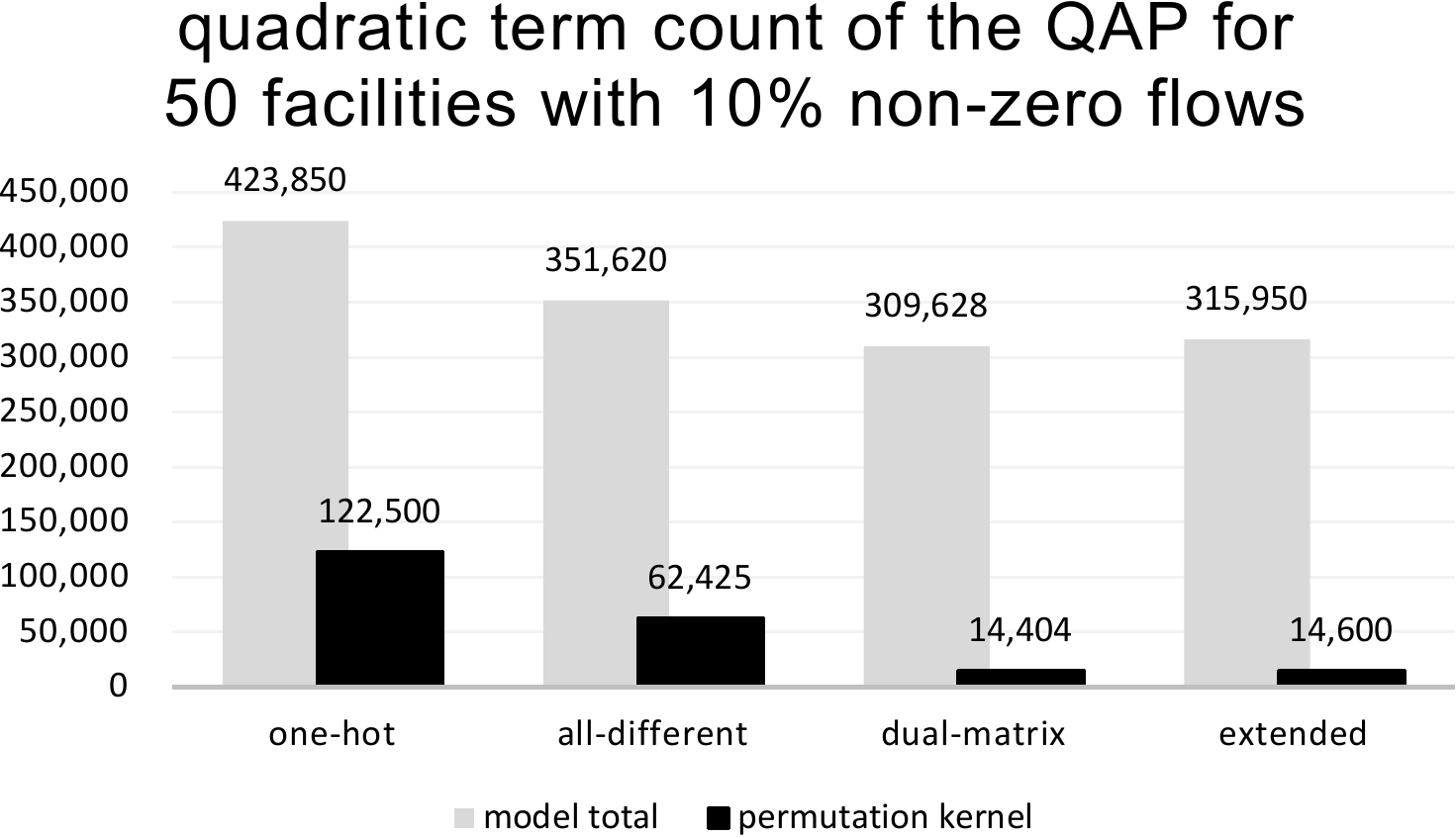}\\
&\\
\includegraphics[width=8cm]{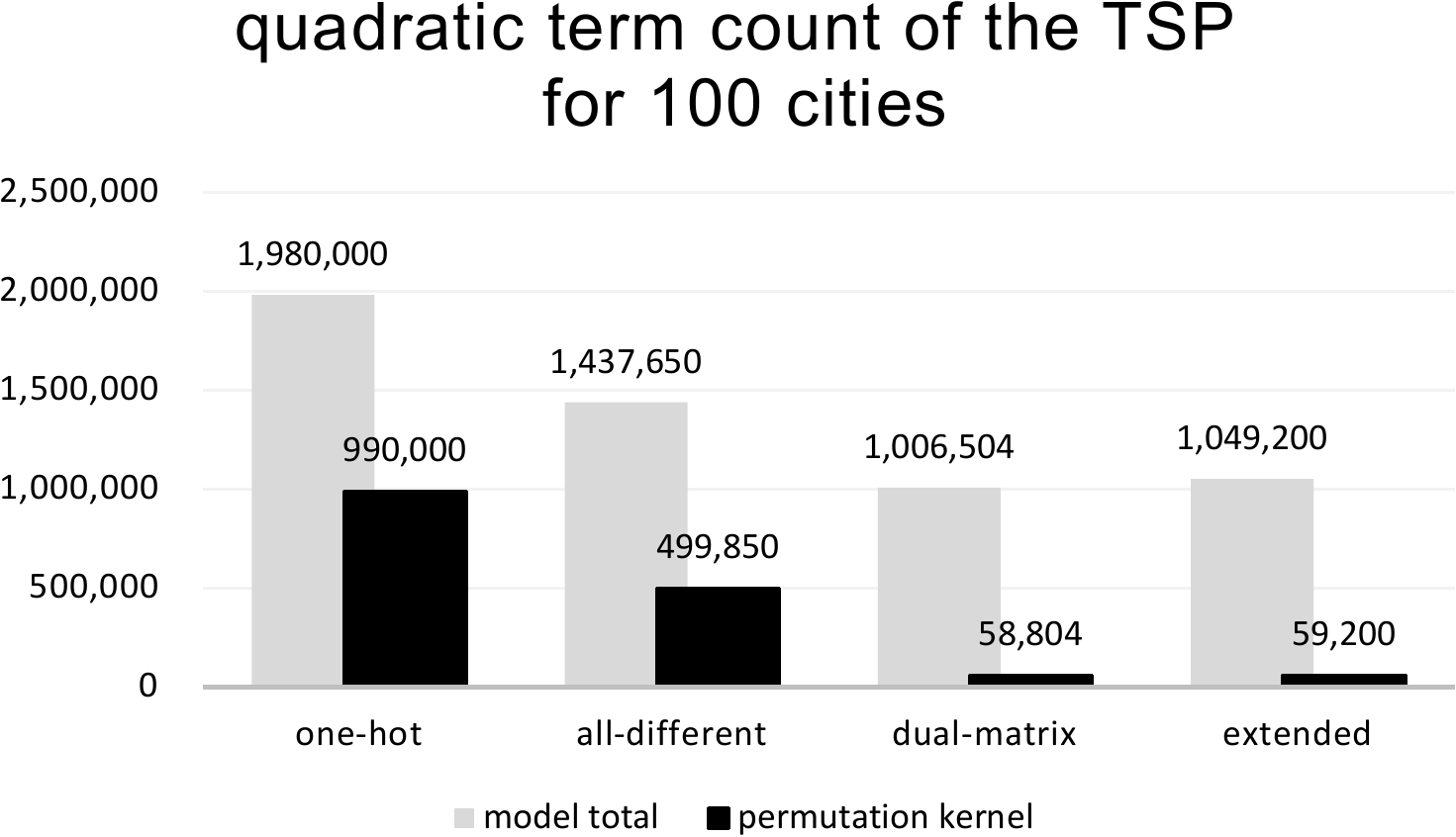}&\includegraphics[width=8cm]{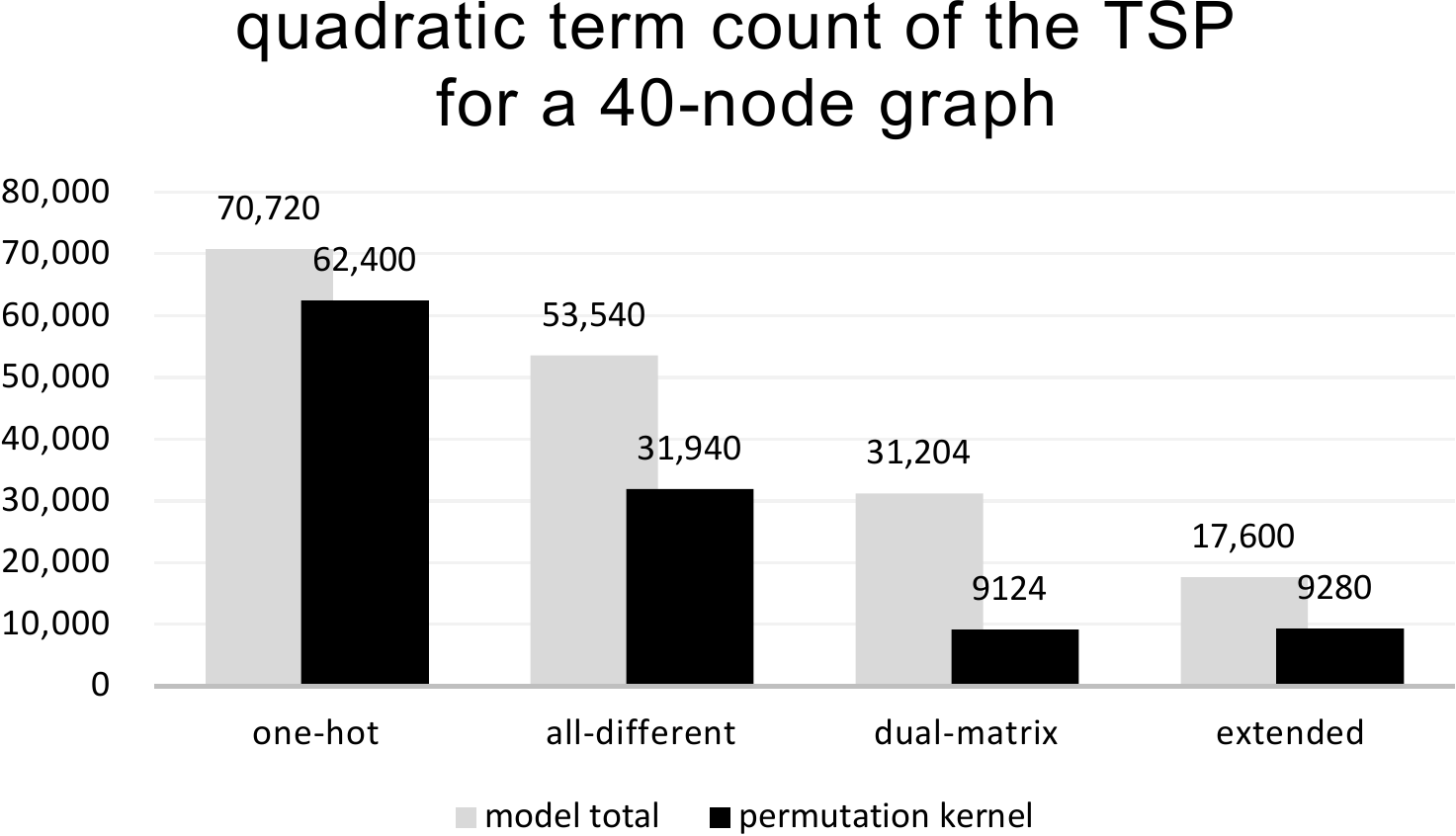}\\
&\\
\includegraphics[width=8cm]{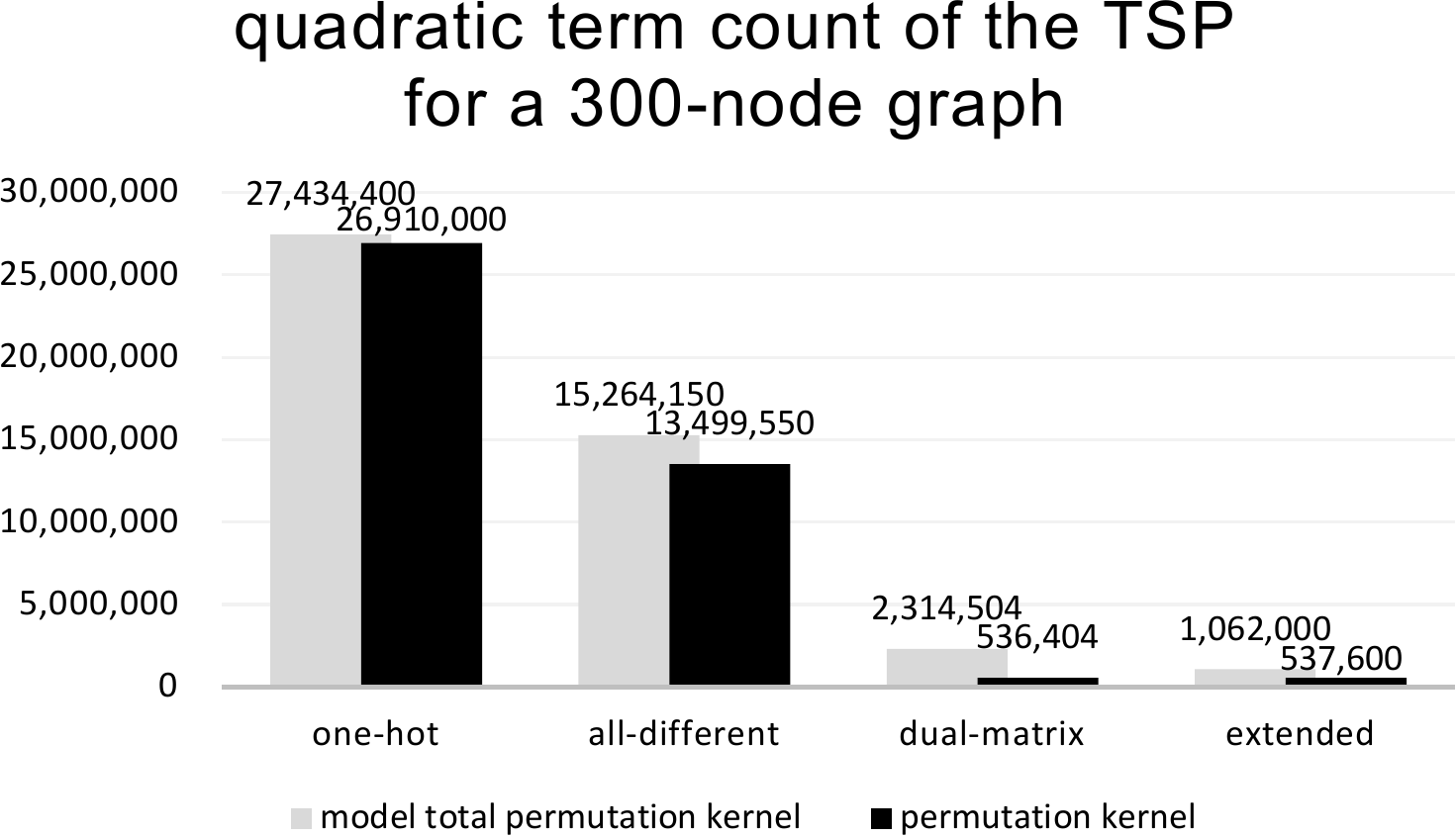} &\includegraphics[width=8cm]{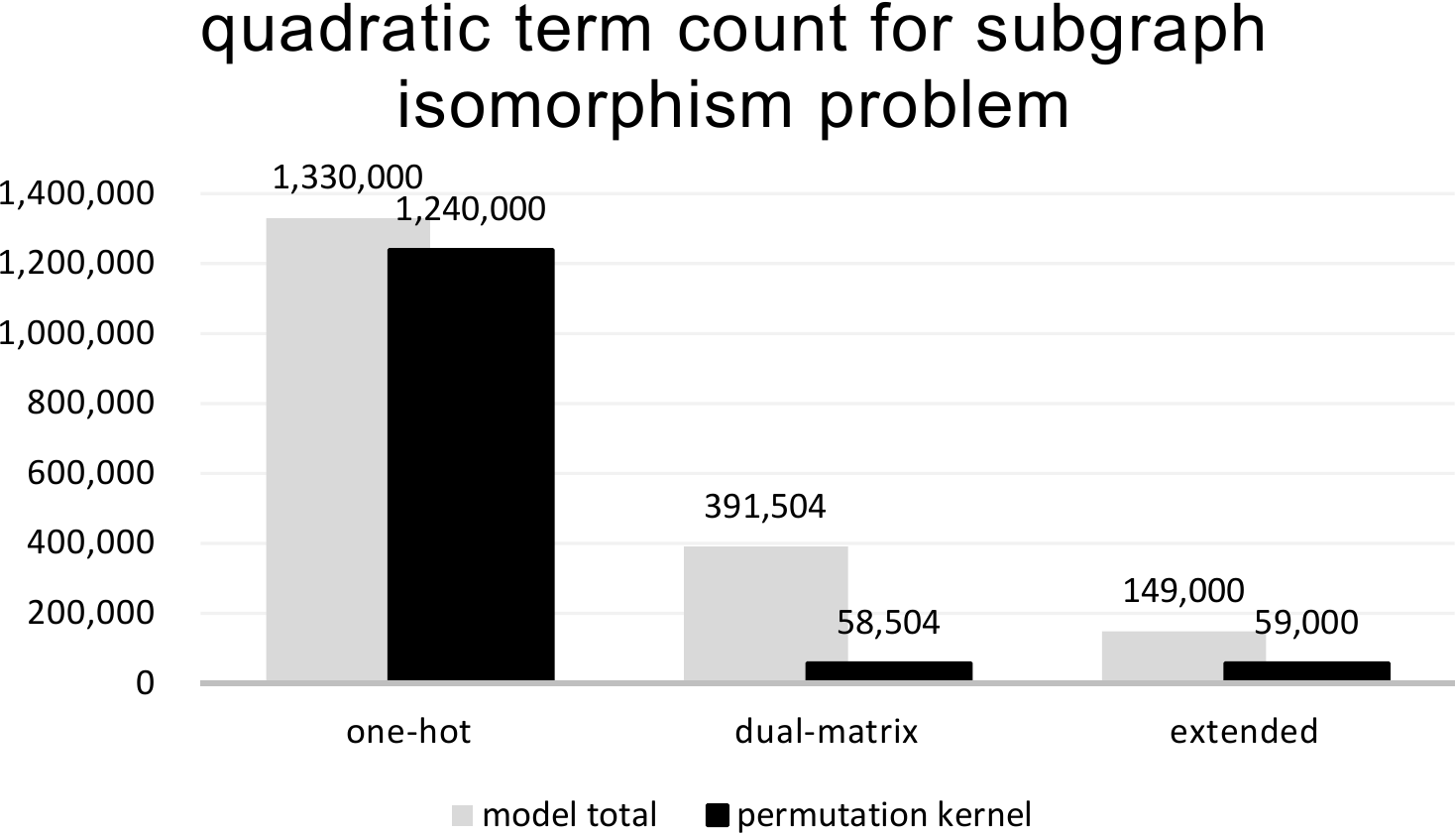} \\
&\\
\includegraphics[width=8cm]{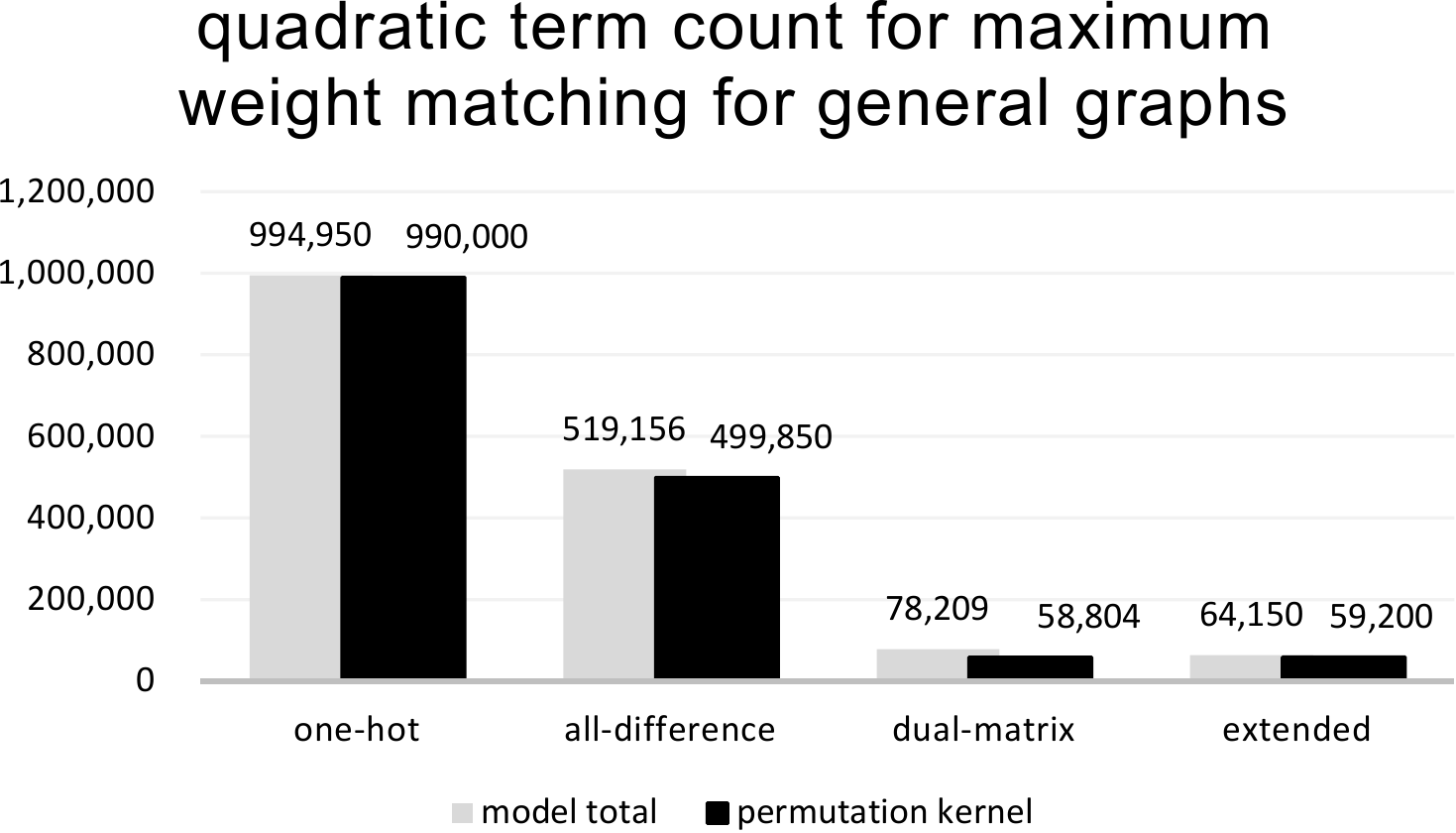}&\includegraphics[width=8cm]{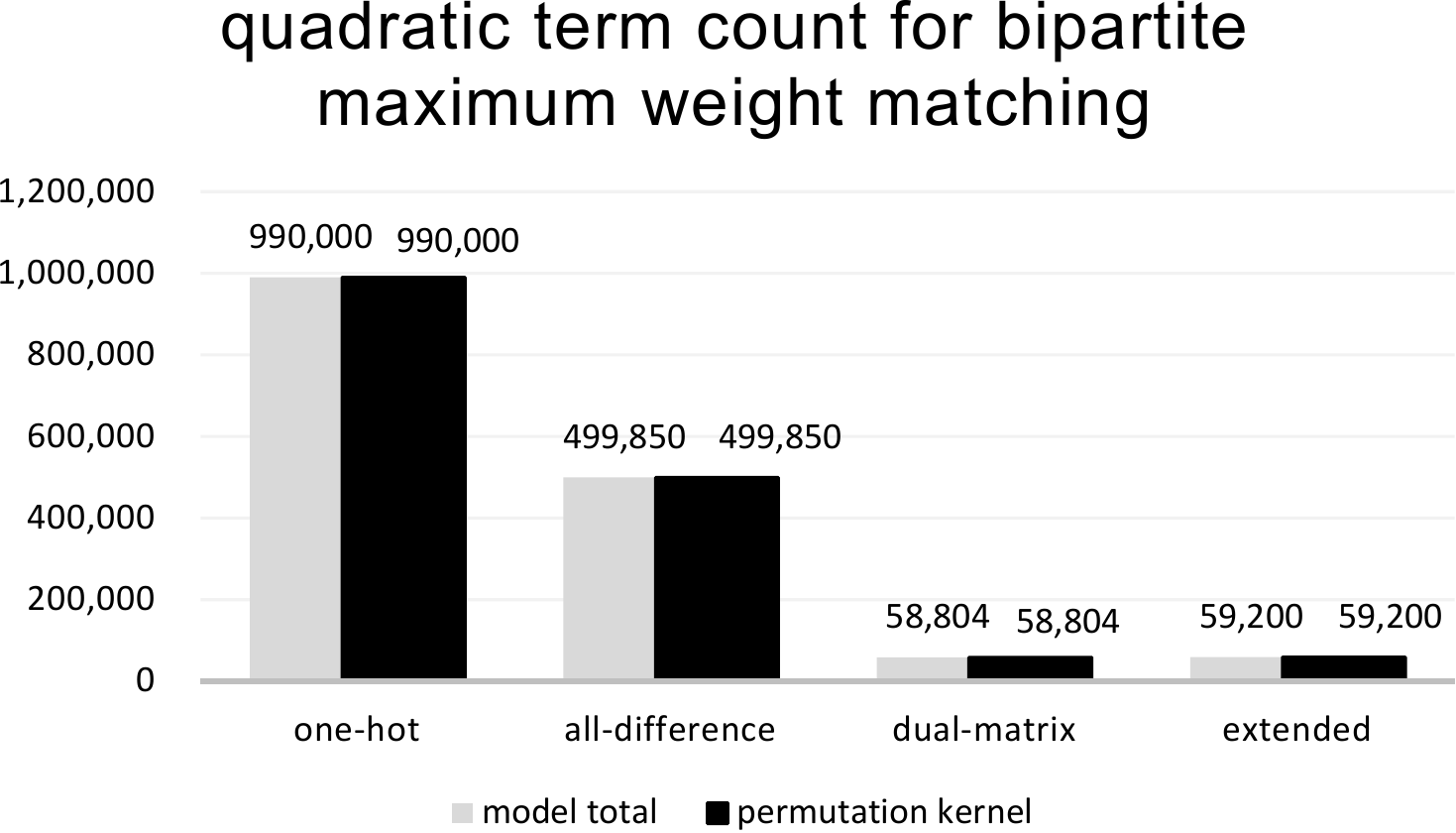}\\
\end{tabular}
	\end{adjustwidth}
\caption{{The quadratic} 
 term counts of Ising models reduced from combinatorial optimization problems
by the conventional one-hot technique (one-hot), the all-different domain-wall technique (all-different),
the dual-matrix domain-wall technique (dual-matrix), and that extended with a one-hot matrix (extended).
We evaluated the total quadratic term counts of the resulting Ising models (model total), as well as those specifically for kernels generating permutations (permutation kernel).}
\label{fig:quadratic_term_count}
\end{figure}

From the figure, it is evident that our dual-matrix domain-wall technique is more effective in reducing the quadratic terms of resulting Ising models when the original PPPs have a lower density.
When the density of the PPP is close to 1, the quadratic terms for PPP interactions outweigh those for permutation kernels.
On the other hand, when the PPP has a lower density, the total number of quadratic terms in the resulting Ising model can be significantly reduced.
For instance,  if the conventional one-hot encoding is used for the TSP with a randomly generated \mbox{300-node} graph, the total number of quadratic terms in the resulting Ising model is 27,434,400.
However, when our dual-matrix domain-wall technique extended by a one-hot matrix is applied, the total number of quadratic terms reduces to only 1,062,000.
This represents a remarkable reduction factor of 25.8.

We should take note of the impact of extending by a one-hot matrix in our dual-matrix domain-wall technique.
Recall that the original technique utilizes variables in matrix $A$ for PPP interactions, while the extended technique introduces an one-hot matrix $X$/$S$.
As shown in Equation (\ref{eq:four}), the dual-matrix domain-wall technique generates four quadratic terms of $A$ for each PPP interaction.
In contrast, the extended dual-matrix domain-wall technique produces only one quadratic term.
When the PPP density is high and involves numerous interactions, the quadratic terms generated by $A$ may overlap, resulting in the Ising model encompassing nearly all pairs of variables in $A$.
Similarly, the Ising models produced by the extended dual-matrix domain-wall technique also include almost all pairs of variables in $X$/$S$.
Consequently, the number of quadratic terms is nearly identical for both techniques, and the extended dual-matrix domain-wall technique does not effectively reduce them.
Likewise, when the PPP density is too low and contains few interactions, the number of quadratic terms generated by variables $A$/$X$/$S$ is small.
Since the extended dual-matrix domain-wall technique yields a slightly larger kernel size, it does not perform well in reducing the overall number of quadratic terms. Therefore, the extended dual-matrix domain-wall technique can effectively reduce the total number of quadratic terms only when the PPP density is moderate---neither too small nor too large.
An example illustrating this is the sub-graph isomorphism problem depicted in Figure \ref{fig:quadratic_term_count}, where the extended dual-matrix domain-wall technique reduces the total number of quadratic terms to less than half.
The quadratic term counts of the extended technique are slightly larger than those of the non-extended technique for the other problem.
However, since the difference is quite small, we should select the extended dual-matrix domain-wall technique if we must use a fixed technique for some reason.

\section{Embedding of Permutation Kernels in D-Wave Quantum Annealer}
\label{sec:embedding}
This section presents the implementation results of Ising kernels on a quantum annealer called the D-Wave Advantage~4.1.
The purpose was to evaluate the suitability of the dual-matrix domain-wall technique for currently available quantum annealers and to verify the impact of quadratic term counts in Ising models on the cell counts.
We employed Ising kernels designed using mathematical expressions such as $H_1^{nn}(S)$ (one-hot), $H_a^{nn}(S)$ (all-different), $H_d^{nn}(S)$ (domain-wall), and 
{$H_e^{nn}(S)$} (extended) for generating permutations. 
Additionally, we used $H_1^{mn}(S)$ (one-hot), $H_d^{mn}(S)$ (domain-wall), and {$H_e^{mn}(S)$} (extended) for partial permutations. 
These Ising kernels were embedded in the D-Wave Advantage~4.1, which was equipped with a 5760-node Pegasus graph designed for solving Ising models.
However, due to faulty nodes and connections, only 5627 cells with 40,279 connections were available.
Therefore, we retrieved the real graph topology and employed “minorminer”, a heuristic tool for finding minor embeddings~\cite{Cai14}, to embed the Ising kernels.

Figure~\ref{fig:implementation} illustrates the number of cells used to embed Ising kernels in the D-Wave Advantage~4.1.
We embedded Ising kernels to generate permutations obtained by $H_1^{nn}(S)$ (one-hot), $H_a^{nn}(S)$ (all-different), {$H_e^{nn}(S)$} (domain-wall), and $H_d^{nn}(S)$ (extended) for $n=5$ until embedding failed. 
For partial permutations, we fixed the parameter $m=12$ and evaluated embedding for $H_1^{mn}(S)$ (one-hot), $H_d^{mn}(S)$ (domain-wall), and {$H_e^{mn}(S)$} (extended) for $n=12$ until embedding failed. 
Ising kernels for generating permutations using the one-hot, all-different, domain-wall, and extended techniques were successfully embedded in the D-Wave Advantage~4.1 topology until $n=16, 28, 35$, and 29, respectively.
Regarding Ising kernels for partial-permutation generation by the one-hot, domain-wall, and extended techniques, successful embedding was achieved until $n=22, 112$, and 75, respectively.
Our dual-matrix domain-wall technique significantly reduced the cell count required for permutation generation.
However, the permutation kernels produced by the extended dual-matrix domain-wall technique utilize slightly more cells than those produced by the non-extended version.
Additionally, the total number of cells used by the extended technique can be smaller, because it uses fewer quadratic terms
for PPP interactions.


{If the entire Ising models for solving combinatorial optimization problems are embedded, the upper bound for the value of $n$ must be smaller.}
For such small values of $n$, conventional digital computers are capable of finding optimal solutions for the problems.
Therefore, it does not make sense to utilize the D-Wave Advantage~4.1 for solving such problems.
However, in the future, if a larger number of cells becomes available, the D-Wave Advantage~4.1 could potentially load Ising models reduced from larger combinatorial optimization problems.
In this scenario, it could uncover optimal solutions that are beyond the reach of conventional digital computers.
In such cases, kernels for generating permutations with larger values of $m$ and $n$ would be employed, and our dual-matrix domain-wall technique would be a powerful approach for designing them.
For a quantum annealer with a fixed size, our dual-matrix domain-wall technique enables us to program larger-sized combinatorial optimization problems than the conventional one-hot encoding techniques, which require a cubic number of quadratic terms.

\begin{figure}[H]
\begin{adjustwidth}{-\extralength}{0cm}
\includegraphics[width=16cm]{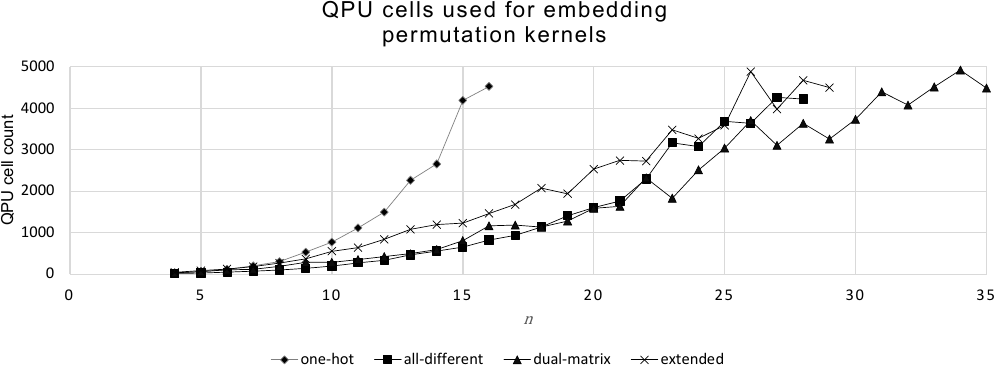}\\
\mbox{}\\
\includegraphics[width=16cm]{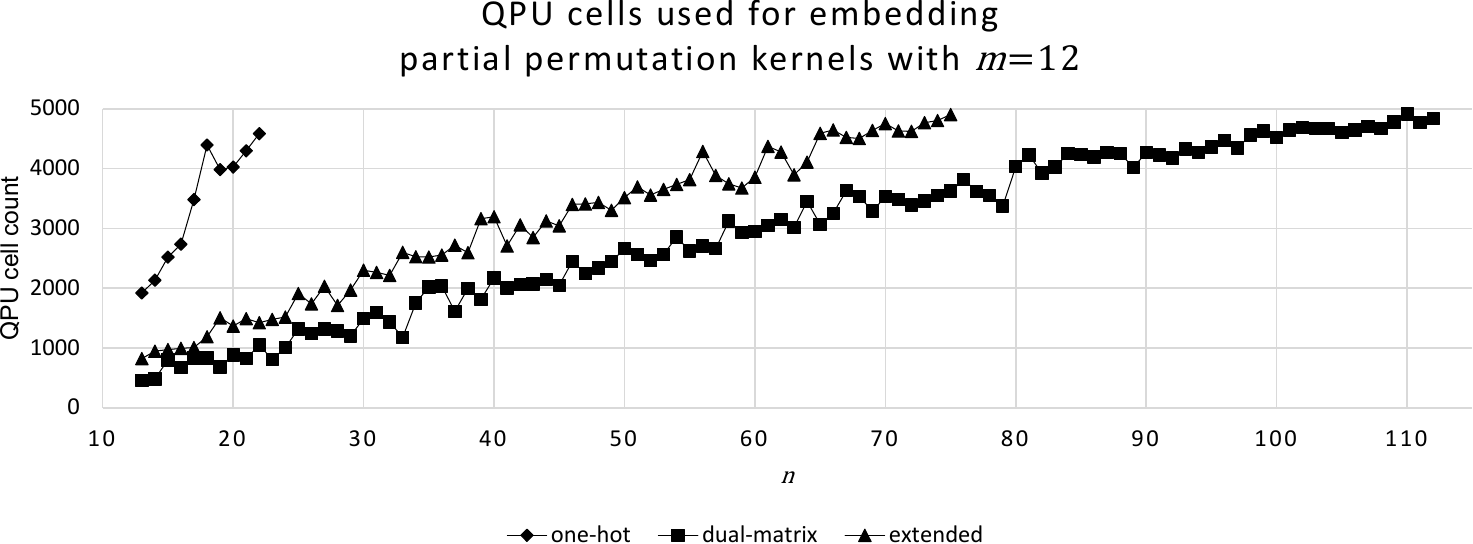}
\end{adjustwidth}
\caption{The numbers of cells of the D-Wave Advantage~4.1 used for embedding permutation kernels.
We evaluated them for Ising kernels that generate permutations of $n$ numbers and partial permutations of $m=12$ numbers selected from $n$ numbers.
 }
\label{fig:implementation}
\end{figure}

\section{Conclusions}\label{sec:concl}
We introduced a novel permutation encoding technique called the dual-matrix domain wall, which can be utilized to design QUBO/Ising kernels for generating permutations and partial permutations.
The conventional one-hot encoding technique for QUBO/Ising kernels requires a cubic number of quadratic terms, whereas our dual-matrix domain-wall technique enables a reduction to a quadratic number of quadratic terms. In this study, we provided a detailed mathematical analysis of QUBO/Ising kernels for permutation generation.

Additionally, we introduced a generic problem known as the particle placement problem (PPP), which can be easily converted to QUBO/Ising models by utilizing the QUBO/Ising kernels for generating permutations. Furthermore, we identified several permutation-based combinatorial optimization problems that can be reduced to the PPP. As a result, these problems can be efficiently converted to QUBO/Ising models using our dual-matrix domain-wall technique.

To evaluate the effectiveness of our approach, we analyzed the number of quadratic terms used in the QUBO/Ising models derived from these optimization problems. Finally, we provided implementation results on the D-Wave Advantage 4.1, a quantum annealer developed by D-Wave Systems.

{
Regrettably, the existing quantum annealers currently accessible are of insufficient size to showcase the practical efficacy of our novel technique.
From a practical perspective, our methodology is specifically tailored for optimal performance exclusively on quantum annealers of a substantial scale.
The dual-matrix domain-wall technique is not presently feasible with the D-Wave Advantage 4.1. Nonetheless, it possesses the potential to evolve into a valuable technique for forthcoming large-scale quantum annealers.
}

\vspace{6pt}

\authorcontributions{Conceptualization, K.N.; methodology, K.N., S.T.. and S.O.; validation, Y.I. and T.Y.; formal analysis, S.T., T.K., R.M., and R.K.; investigation, K.N.,Y.I., J.Y., T.K., and S.O.;  writing---original draft preparation, K.N.; writing---review and editing,  K.N., J.Y., T.K., R.M., and R.K.; visualization,  K.N.; supervision,  K.N.; project administration, T.Y.; funding acquisition, T.Y.
All authors have read and agreed to the published version of the manuscript.}

\funding{This research received no external funding.}

\institutionalreview{{Not applicable.}
}

\informedconsent{{Not applicable.}

}

\dataavailability{{A Python program for generating QUBO/Ising models, as utilized in this paper, is available on GitHub at https://github.com/nakanocs/dual-matrix.}
}

\conflictsofinterest{{Koji Nakano, Takashi Yazane, Junko Yano, Takumi Kato, Shiro Ozaki, Rie Mori, and Ryota Katsuki are the inventors listed on Japanese Patent Application No. 2023-125848. This manuscript discusses the inventions outlined in this patent application.}} 

\begin{adjustwidth}{-\extralength}{0cm}
\reftitle{References}

\PublishersNote{}
\end{adjustwidth}
\end{document}